\documentclass[12pt,a4paper]{article}
\usepackage[utf8]{inputenc}
\usepackage[english]{babel}
\usepackage{amsmath}
\usepackage{amsthm}
\usepackage{amsfonts}
\usepackage{amssymb}
\usepackage{geometry}
\usepackage[pdftex]{graphicx}
\usepackage[longnamesfirst,authoryear]{natbib}
\usepackage{afterpage}

\newcommand{\indepe}{\mathop{\perp\!\!\!\perp}}

\newtheorem{Theorem}{Theorem}
\newtheorem{Lemma}{Lemma}
\newtheorem{Proposition}{Proposition}

\newtheorem{Assumption}{Assumption}

\newtheorem{Example}{Example}
\newtheorem{Simulation}{Simulation}
\newtheorem{Remark}{Remark}

\begin{document}

\author{Takuya Ishihara\footnote{Waseda University, Faculty of Social Sciences, email: takuya319ti@gmail.com. I would like to express my appreciation to the editor and the anonymous referees for their careful reading and comments on the paper. I also would like to thank Katsumi Shimotsu, Hidehiko Ichimura, and the seminar participants at the University of Tokyo, Otaru University of Commerce, Kanazawa University, Hiroshima University, and Shanghai Jiao Tong University. This work was supported by the Grant-in-Aid for JSPS Research Fellowship (17J03043) from the JSPS.}}

\title{Partial Identification of Nonseparable Models using Binary Instruments}

\date{\today}
\maketitle

\begin{abstract}
In this study, we explore the partial identification of nonseparable models with continuous endogenous and binary instrumental variables.
We show that the structural function is partially identified when it is monotone or concave in the explanatory variable.
\cite{d2015identification} and \cite{torgovitsky2015identification} prove the point identification of the structural function under a key assumption that the conditional distribution functions of the endogenous variable for different values of the instrumental variables have intersections.
We demonstrate that, even if this assumption does not hold, monotonicity and concavity provide identifying power.
Point identification is achieved when the structural function is flat or linear with respect to the explanatory variable over a given interval.
We compute the bounds using real data and show that our bounds are informative.
\if0
\fi
\end{abstract}

%%%%%%%%%%%%%%%%%%%%%%%%%%%%%%%%%%%%%%%%%%%
\section{Introduction}

In this study, we examine the identification of a system of structural equations that takes the following form:
\begin{equation}
\label{model}
\begin{aligned}
Y &=& g(X,\epsilon) \\
X &=& h(Z,\eta),
\end{aligned}
\end{equation}
where $Y\in\mathbb{R}$ is a scalar response variable, $X\in\mathbb{R}$ is a continuous endogenous variable, $Z \in \{0,1\}$ is a binary instrument, and $\epsilon$ and $\eta$ are unobservable scalar variables.
This specification is nonseparable in the unobservable variable $\epsilon$ and captures the unobserved heterogeneity in the effect of $X$ on $Y$.
Such models have also been considered by, for example, \cite{d2015identification} and \cite{torgovitsky2015identification}.
For any random variable $U$ and random vector $W$, let $F_{U|W}(u|w)$ denote the conditional distribution function of $U$ conditional on $W$. 
In some places, we interchangeably use the notation $F_{U|W=w}(u)$ instead of $F_{U|W}(u|w)$.
Let $\mathcal{X}$, $\mathcal{X}_z$, and $\mathcal{Y}_{x,z}$ denote the interiors of the support of $X$, $X|Z=z$, and $Y|X=x,Z=z$, respectively.

\cite{d2015identification} and \cite{torgovitsky2015identification} show that $g$ is point identified when $g(x,e)$ and $h(z,v)$ are strictly increasing in $e$ and $v$ and $Z$ is independent of $(\epsilon,\eta)$.
Their results are important for empirical analyses in which many instruments are binary or discrete, such as the intent-to-treat in a randomized controlled experiment or quarter of birth used by \cite{angrist1991does}.
For nonparametric models with a continuously distributed $X$, several point identification results require $Z$ to be continuously distributed.
See, for example, \cite{newey1999nonparametric} and \cite{imbens2009identification}.

\cite{d2015identification} and \cite{torgovitsky2015identification} assume that $F_{X|Z}(x|0)$ and $F_{X|Z}(x|1)$ have intersections, when establishing point identification for $g$.
However, many empirically important models do not satisfy this assumption.
For example, $F_{X|Z}(x|0)$ and $F_{X|Z}(x|1)$ do not have an intersection when $Z$ has a strictly monotonic effect on $X$ such as linear models $X=\beta_0 + \beta_1 Z + \eta$.
Further, in many applications, instrumental variables have a strictly monotonic effect on endogenous variables (e.g. the LATE framework proposed by \cite{imbens1994identification}).
For example, as in \cite{macours2012cash}, cash transfer programs have been implemented in several countries.
As such, if we use treatment indicator $Z$ as the instrumental variable for income $X$, $Z$ has a strictly monotonic effect on $X$, which violates the intersection assumption.
Hence, $F_{X|Z}(x|0)$ and $F_{X|Z}(x|1)$ never have an intersection in this example.
Actually, in Section 5, we show that $F_{X|Z}(x|0)$ and $F_{X|Z}(x|1)$ do not have an intersection in the real data.

This study shows that, when $g(x,e)$ is monotone or concave in $x$, we can partially identify $g$, even if $F_{X|Z}(x|0)$ and $F_{X|Z}(x|1)$ have no intersection.
The structural function $g(x,e)$ is monotone or concave in $x$ in many economic models.
For example, the demand function is decreasing in price if the income effect is negligible, and economic analyses of production often suppose that the production function is monotone and concave in inputs.
In general, the demand function is not decreasing in price.
For instance, \cite{hoderlein2011many} employs nonseparable models and analyzes consumer behavior without the monotonicity assumption.
Many studies employ monotonicity or concavity to identify the target parameters (e.g., \cite{manski1997monotone}, \cite{giustinelli2011non}, \cite{d2013nonlinear}, and \cite{okumura2014concave}).
Specifically, \cite{manski1997monotone} imposes these assumptions and shows that the average treatment response is partially identified.
The partial identification approach using the concavity assumption in this study is somewhat similar to that considered by \cite{d2013nonlinear}.

In this model, monotonicity and concavity provide identifying power.
\cite{d2015identification} and \cite{torgovitsky2015identification} show that when $F_{X|Z}(x|0)$ and $F_{X|Z}(x|1)$ have intersections, $T_{x',x}(y) \equiv g(x',g^{-1}(x,y))$ is identified for all $x$, $x'$, and $y$, where $g^{-1}(x,y)$ is the inverse of $g$ with respect to its last component.
Then, $g$ is point identified under appropriate normalization.
By contrast, when $F_{X|Z}(x|0)$ and $F_{X|Z}(x|1)$ do not have intersections, we only identify $T_{x',x}(y)$ for some $x$ and $x'$.
Although this information restricts the functional form of $g$, it does not provide the informative bounds of $g$.
In this case, monotonicity and convexity allow us to interpolate or extrapolate $T_{x',x}(y)$ and provide the informative bounds of $T_{x',x}(y)$.
For example, if $T_{x',x}(y)$ is identified and $\tilde{x} \geq x'$, monotonicity implies $T_{\tilde{x},x}(y) \geq T_{x',x}(y)$, and hence, we obtain a lower bound of $T_{\tilde{x},x}(y)$.
Using these bounds, we can achieve the partial identification of $g$.

There is a rich literature on the identification of nonseparable models using the control function approach.
For example, \cite{chesher2007instrumental}, \cite{hoderlein2007identification}, \cite{florens2008identification}, \cite{imbens2009identification}, \cite{hoderlein2009identification}, \cite{hoderlein2011many}, \cite{kasy2011identification}, and \cite{blundell2013control} consider the identification of nonseparable models using the control function approach.
Particularly, \cite{imbens2009identification} consider models similar to (\ref{model}).
Their study allows $\epsilon$ to be multivariate, showing that the quantile function of $g(x,\epsilon)$ is point identified, while in this analysis, $\epsilon$ is imposed as scalar.
Their results need continuous instruments, whereas those of \cite{d2015identification}, \cite{torgovitsky2015identification}, and the present study do not.

We assume that the instrumental variable $Z$ is binary.
\cite{d2015identification} consider the case in which the instrumental variable takes more than two values, thus showing point identification can be achieved using group and dynamical systems theories even when $F_{X|Z}(x|z)$ and $F_{X|Z}(x|z')$ have no intersection.

\cite{caetano2017identifying} provides alternative results for the identification of nonseparable models with continuous endogenous variables and binary instruments.
To this end, they use the observed covariates to identify the structural function.
Although their approach does not require $F_{X|Z}(x|z)$ and $F_{X|Z}(x|z')$ to intersect, they assume the structural function does not depend on the observed covariates.
By contrast, our identification approach does not require the existence of covariates and allows the structural function to depend on the observed covariates.\footnote{For simplicity, we consider the case where there are no covariates. It is thus straightforward to extent our model to the model with covariates.}

The remainder of this study is organized as follows.
Section 2 introduces the assumptions employed in the analysis.
Section 3 demonstrates our partial identification strategy and shows that we cannot identify $g$ without any shape restrictions.
Sections 4 provides the lower and upper bounds of $g$ under the monotonicity and concavity assumptions.
Section 5 computes the bounds using real data.
Section 6 concludes the paper.

%%%%%%%%%%%%%%%%%%%%%%%%%%%%%%%%%%%%%%%%%%%
\section{Model}

The following two assumptions are the same as those in \cite{d2015identification} and \cite{torgovitsky2015identification}:

%%% Assumption 1 %%%
\begin{Assumption}
The instrument is independent of the unobservable variables: $Z\indepe(\epsilon,\eta)$.
\end{Assumption}

%%% Assumption 2 %%%
\begin{Assumption}
(i) The function $g$ is continuous and $g(x,e)$ is strictly increasing in $e$ for $x \in \mathcal{X}$. (ii) For $z \in \{0,1\}$, $h(z,v)$ is continuous and strictly increasing in $v$. 
\end{Assumption}

Assumptions 1 and 2 (ii) are typically employed when using the control function approach.
See, for example, \cite{imbens2009identification}, \cite{d2015identification}, and \cite{torgovitsky2015identification}.
Although Assumption 2 (i) is strong, it is necessary for our identification approach.
\cite{hoderlein2007identification}, \cite{hoderlein2009identification}, \cite{hoderlein2011many}, and \cite{imbens2009identification} do not employ this assumption.

The next assumption regarding the conditional distributions of $X$ conditional on $Z$ differs from that of \cite{d2015identification} and \cite{torgovitsky2015identification}.

%%% Assumption 3 %%%
\begin{Assumption}
(i) The conditional distribution $F_{X|Z}(x|z)$ is continuous in $x$ for $z \in \{0,1\}$ and $F_{X|Z}(x|0) < F_{X|Z}(x|1)$ for $x \in \mathcal{X}$. (ii) We have $\mathcal{X}_0 = (\underline{x}_0,\overline{x}_0)$, $\mathcal{X}_1 = (\underline{x}_1,\overline{x}_1)$, and $-\infty < \underline{x}_1<\underline{x}_0 < \overline{x}_1<\overline{x}_0 < \infty$.
\end{Assumption}

Conditions (i) and (ii) above imply that $F_{X|Z}(x|z)$ is strictly increasing and continuous in $x$ conditional on $\mathcal{X}_{z}$.
Further, condition (i) implies that $F_{X|Z}(x|0)$ and $F_{X|Z}(x|1)$ do not have any intersection on the support of $X$ and $X|Z=0$ stochastically dominates $X|Z=1$.
Therefore, $Z$ has a strictly monotonic effect on $X$.
\cite{d2015identification} and \cite{torgovitsky2015identification} rule out this case because they assume $F_{X|Z}(x|0)$ and $F_{X|Z}(x|1)$ have intersections on the support of $X$.
Condition (ii) implies that $\underline{x}_0 \neq \underline{x}_1$ and $\overline{x}_0 \neq \overline{x}_1$, which may be restrictive in some cases.
For example, \cite{torgovitsky2015identification} considers an experiment that randomly assigns students across various schools to a large or small class ($Z = 0$ or $1$, respectively).
Then, he shows that $\underline{x}_0 = \underline{x}_1$ can happen when $X$ is the class-size, $Z$ is the randomly assigned intent-to-treat, and partial compliance arises.

When we have $\mathcal{X}_0=\mathcal{X}_1=(\underline{x},\overline{x})$, then $F_{X|Z}(x|0)$ and $F_{X|Z}(x|1)$ must have intersections at the boundary points of the support of $X$.
However, in this case, $g$ is not identified unless $g(\underline{x},e)$ (or $g(\overline{x},e)$) exists and $g(\underline{x},e)$ (or $g(\overline{x},e)$) is strictly increasing in $e$.
\cite{torgovitsky2015identification} shows that the point identification of $g$ holds when $F_{X|Z}(x|0)$ and $F_{X|Z}(x|1)$ intersect at a boundary point $\underline{x}$, and $g(\underline{x},e)$ exists and is strictly increasing in $e$.

Next, we impose restrictions on the conditional distributions of $Y$ conditional on $X$ and $Z$.

%%% Assumption 4 %%%
\begin{Assumption}
(i) For $(z,x,y) \in \{0,1\} \times \mathcal{X}_z \times \mathcal{Y}_{x,z}$, $F_{Y|X,Z}(y|x,z)$ is continuous in $x$ and $y$. (ii) For $(z,x) \in \{0,1\} \times \mathcal{X}_z$, we have $\mathcal{Y}_{x,z} = \mathcal{Y} = (\underline{y},\overline{y})$, where $- \infty \leq \underline{y}<\overline{y} \leq \infty$.
\end{Assumption}

\cite{d2015identification} and \cite{torgovitsky2015identification} also assume condition (i) but not condition (ii).
Both conditions imply that $F_{Y|X,Z}(y|x,z)$ is strictly increasing and continuous in $y$ on $\mathcal{Y}$.
Hence, the conditional quantile function of $Y$ conditional on $X$ and $Z$ is the inverse of $F_{Y|X,Z}(y|x,z)$.
Condition (ii) is not necessary for this study's results but, without it, deriving the results can become cumbersome.
In Appendix 3, we derive the bounds of $g$ without this condition.

Finally, we impose the normalization assumption on unobservable variables and support condition of $\epsilon|X=x,Z=z$.

%%% Assumption 5 %%%
\begin{Assumption}
(i) We have $\epsilon \sim U(0,1)$ and $\eta \sim U(0,1)$. (ii) For $(z,x) \in \{0,1\} \times \mathcal{X}_z$, the interior of the support of $\epsilon|X=x,Z=z$ is $(0,1)$.
\end{Assumption}

Condition (i) is the usual normalization in a nonseparable model (see \cite{matzkin2003nonparametric}).
\cite{torgovitsky2015identification} does not use this normalization, while \cite{d2015identification} normalize $\epsilon$ to be uniformly distributed.
Condition (ii) implies that $g(x,e) \in (\underline{y},\overline{y}) = \mathcal{Y}$ for all $(x,e) \in \mathcal{X} \times (0,1)$.
Condition (ii) is necessary because, if the support of $\epsilon|X=x,Z=z$ is $[0,\bar{e}]$ for some $0<\bar{e}<1$, then the conditional support of $Y$ given $X=x$ and $Z=z$ is equal to $\{g(x,e):e \in [0,\bar{e}]\}$ and we have $g(x,e) \not\in \mathcal{Y}$ for $e>\bar{e}$.
This implies that we can not identify $g(x,e)$ for $e>\bar{e}$.

%%%%%%%%%%%%%%%%%%%%%%%%%%%%%%%%%%%%%%%%%%%%%
\begin{Example}[Cash Transfer Programs]
Cash transfer programs have been conducted in many countries and many papers estimate their impacts on early childhood development by using randomized experiments.
For example, \cite{macours2012cash} analyze the impact of a cash transfer program on early childhood cognitive development.
In this program, participants were randomly assigned to either the treatment or control groups.
As such, we can consider the following model:
\begin{eqnarray}
Y &=& g(X,\epsilon), \nonumber \\
X &=& \tilde{Z}h_1(\eta) + (1-\tilde{Z})h_0(\eta), \nonumber
\end{eqnarray}
where $Y$ is the child's outcome of cognitive development, $X$ is the total expenditure, and $\tilde{Z}$ is the treatment indicator of the program.
Because cash transfers usually increase total expenditure, we can assume $h_1(\eta)-h_0(\eta)>0$.
When participants are randomly assigned to either the treatment or control groups, $Z \equiv 1-\tilde{Z}$ is independent of $(\epsilon,\eta)$ and hence Assumption 1 is satisfied.
Because $Z$ is independent of $\eta$, we have $F_{X|Z}(x|1) = P(h_0(\eta) \leq x)$ and $F_{X|Z}(x|0) = P(h_1(\eta) \leq x)$.
Since $h_1(\eta) > h_0(\eta)$, we have $F_{X|Z}(x|0) < F_{X|Z}(x|1)$ for all $x$.
In this case, Assumption 3 is satisfied, that is, $F_{X|Z}(x|0)$ and $F_{X|Z}(x|1)$ have no intersection.
In Section 5, we show this assumption actually holds for the data used by \cite{macours2012cash}.
\end{Example}
%%%%%%%%%%%%%%%%%%%%%%%%%%%%%%%%%%%%%%%%%%%

%%%%%%%%%%%%%%%%%%%%%%%%%%%%%%%%%%%%%%%%%%%
\section{Basic idea of identification}

In this section, we explain the basic idea of our identification approach.
Let $\overline{\mathcal{Y}}$ be the closure of $\mathcal{Y}$.
We establish the partial identification of $g$ by showing we can identify functions $T^U_{x',x}(y):\overline{\mathcal{Y}}\rightarrow \overline{\mathcal{Y}}$ and $T^L_{x',x}(y):\overline{\mathcal{Y}}\rightarrow \overline{\mathcal{Y}}$ and they are (i) strictly increasing in $y$, (ii) surjective, that is, $T^U_{x',x}\left( [\underline{y}, \overline{y}] \right)=T^L_{x',x}\left( [\underline{y}, \overline{y}] \right)=[\underline{y}, \overline{y}]$, and (iii) satisfy the following inequalities:
\begin{eqnarray}
g(x',e) &\leq & T_{x',x}^U \left( g(x,e) \right), \label{TU} \\
g(x',e) &\geq & T_{x',x}^L \left( g(x,e)\right). \label{TL} 
\end{eqnarray}
From (\ref{TU}) and (\ref{TL}), $T_{x',x}^U(y)$ and $T_{x',x}^L(y)$ are the upper and lower bounds of $T_{x',x}(y) \equiv g(x',g^{-1}(x,y))$, respectively.
If $T_{x',x}^U(y)$ is identified for all $x,x' \in \mathcal{X}$, we can obtain the lower bound of the structural function $g(x,e)$ in the following manner.
Here, we define $G_x^L(u) \equiv \int F_{Y|X=x'} \left( T_{x',x}^U (u) \right) dF_{X}(x')$.
If $T_{x',x}^U(y)$ satisfying (\ref{TU}) is obtained for all $x,x' \in \mathcal{X}$, then we have
\begin{eqnarray}
 G_x^L\left( g(x,e) \right) &=&\int F_{Y|X=x'} \left( T_{x',x}^U \left(g(x,e)\right) \right) dF_{X}(x') \nonumber \\ 
 &\geq & \int F_{Y|X=x'}\left(g(x',e)\right) dF_{X}(x') \nonumber \\
 &=& \int P(g(x',\epsilon) \leq g(x',e)|X=x') dF_{X}(x') \nonumber \\
 &=& \int P(\epsilon \leq e |X=x') dF_{X}(x') \ = \ e, \label{approach_DF1}
\end{eqnarray}
where the first inequality follows from (\ref{TU}) and the third equality follows from the strict monotonicity of $g(x,e)$ in $e$.
Because $F_{Y|X=x'}(y)$ and $T_{x',x}^U(y)$ are strictly increasing in $y$, $u<u'$ implies $F_{Y|X=x'} \left( T_{x',x}^U (u) \right) < F_{Y|X=x'} \left( T_{x',x}^U (u') \right)$ for all $x'$.
Hence, $G_x^L(u)$ is strictly increasing in $u$.
Because $T_{x',x}^U(y)$ is surjective, we have $G_x^L\left( [\underline{y},\overline{y}] \right) = [0,1]$.
Hence, for all $e \in (0,1)$, we have
\begin{equation}
g(x,e) \geq \left(G_x^L\right)^{-1}(e). \label{approach_DF2}
\end{equation}
Similarly, we define $G_x^U(u) \equiv \int F_{Y|X=x'} \left( T_{x',x}^L (u) \right) dF_{X}(x')$, and thus, we have $g(x,e) \leq \left(G_x^U\right)^{-1}(e)$.
These bounds are pointwise not uniform.
Throughout this paper, we focus on pointwise bounds of $g$.

Next, we explain how to construct functions $T^U_{x',x}(y)$ and $T^L_{x',x}(y)$ that satisfy (\ref{TU}) and (\ref{TL}).
For any random variable $U$ and random vector $W$, let $Q_{U|W}(\tau|w)$ denote the conditional $\tau$-th quantile of $U$ conditional on $W=w$, that is, $Q_{U|W}(\tau|w) \equiv \inf\{u:F_{U|W}(u|w) \geq \tau \}$.
As in \cite{torgovitsky2015identification}, we define $\pi(x):\mathcal{X}_0\rightarrow\mathcal{X}_1$ and $\pi^{-1}(x):\mathcal{X}_1\rightarrow\mathcal{X}_0$\footnote{These functions correspond to $s_{ij}$ in \cite{d2015identification}.} as:
\begin{equation}\label{pi}
\begin{aligned}
 \pi(x) &\equiv & Q_{X|Z}\left( F_{X|Z}(x|0) | 1 \right), \\
 \pi^{-1}(x) &\equiv & Q_{X|Z}\left( F_{X|Z}(x|1) | 0 \right).
\end{aligned}
\end{equation}
Figure 1 illustrates functions $\pi(x)$ and $\pi^{-1}(x)$.
By definition of $\pi(x)$, if $x \not\in \mathcal{X}_0$, then $\pi(x)$ does not exist.
Similarly, if $x \not\in \mathcal{X}_1$, then $\pi^{-1}(x)$ does not exist.
The following result is essentially proven by \cite{d2015identification} (Theorem 1).
However, we state this result as a proposition because it plays a central role in the following and our assumptions differ somewhat from those of \cite{d2015identification}.

%%% Proposition 1 %%%
\begin{Proposition}
We define
\begin{eqnarray}
\tilde{T}_{x,1}(y) &\equiv & Q_{Y|X,Z} \left( F_{Y|X,Z} \left( y | x,0 \right)|\pi(x), 1 \right) \ \ \ \ \text{for $x \in \mathcal{X}_0$  and} \nonumber \\
\tilde{T}_{x,-1}(y) &\equiv & Q_{Y|X,Z} \left( F_{Y|X,Z} \left( y|x,1 \right)|\pi^{-1}(x),0 \right) \ \ \ \ \text{for $x \in \mathcal{X}_1$.} \nonumber
\end{eqnarray}
Then, under Assumptions 1--5, we have
\begin{eqnarray}
 g\left( \pi(x),e \right) &=& \tilde{T}_{x,1} \left( g(x,e) \right) \ \ \ \ \text{for $x \in \mathcal{X}_0$  and} \nonumber \\
 g\left( \pi^{-1}(x),e \right) &=& \tilde{T}_{x,-1} \left( g(x,e) \right) \ \ \ \ \text{for $x \in \mathcal{X}_1$.} \nonumber
\end{eqnarray}
\end{Proposition}

We provide the sketch of proof.
We define 
\begin{equation}
V \equiv F_{X|Z}(X|Z). \label{def_V}
\end{equation}
This is called ``control variable'' in \cite{imbens2009identification}.
From Assumptions 1 and 5 (i), we obtain $V=\eta$.
Because $\{X=x,Z=z\}$ is equivalent to $\{F_{X|Z}(X|Z)=(x|z), Z=z\}$, we have
$$
F_{\epsilon|X,Z}(e|x,z) = F_{\epsilon|V,Z}(e|F_{X|Z}(x|z),z).
$$
Because the control variable $V$ is equal to $\eta$, it follows from Assumption 1 that
\begin{eqnarray}
P(\epsilon \leq e | X=x, Z=0) = P(\epsilon \leq e | X=\pi(x), Z=1). \label{control function}
\end{eqnarray}

Next, we show that (\ref{control function}) implies $g\left( \pi(x),e \right) = \tilde{T}_{x,1} \left( g(x,e) \right)$.
It follows from (\ref{control function}) and the strict monotonicity of $g$ that 
\begin{eqnarray}
F_{Y|X,Z}(g(x,e)|x,0) &=& P\left( g(x,\epsilon) \leq g(x,e) | X=x,Z=0 \right) \nonumber \\
&=& P\left( \epsilon \leq e | X=x,Z=0 \right) \nonumber \\
&=& P\left( \epsilon \leq e | X=\pi(x),Z=1 \right) \nonumber \\
&=& F_{Y|X,Z}(g(\pi(x),e)|\pi(x),1). \nonumber 
\end{eqnarray}
Hence, we obtain $g\left( \pi(x),e \right) = \tilde{T}_{x,1} \left( g(x,e) \right)$.
Similarly, we also obtain $g\left( \pi^{-1}(x),e \right) = \tilde{T}_{x,-1} \left( g(x,e) \right)$ and we prove Proposition 1.

By definition, $\tilde{T}_{x,1}(y)$ is strictly increasing in $y$ for $x \in \mathcal{X}_0$, $\tilde{T}_{x,-1}(y)$ is strictly increasing in $y$ for $x \in \mathcal{X}_1$, $\tilde{T}_{x,1}([\underline{y},\overline{y}]) = [\underline{y},\overline{y}]$, and $\tilde{T}_{x,-1}([\underline{y},\overline{y}]) = [\underline{y},\overline{y}]$.
For $n \in \mathbb{N}$, we define $\pi^{n}(x)$ and $\pi^{-n}(x)$ as the follows:
\begin{eqnarray}
& \pi^0(x) \equiv  x  \ \ \ \ \ \ \ \ \ \ \ \ \ \ \ \ \ \ & \text{ for all $x \in \mathcal{X}$}, \nonumber \\
& \pi^n(x)  \equiv  \pi \circ \pi^{n-1}(x) \ \ \ \ \ & \text{ if $\pi^{n-1}(x) \in \mathcal{X}_0$,} \nonumber \\
& \pi^{-n}(x) \equiv  \pi^{-1} \circ \pi^{-(n-1)}(x) & \text{ if $\pi^{-(n-1)}(x) \in \mathcal{X}_1$.} \nonumber
\end{eqnarray}
Because the domain of $\pi$ is $\mathcal{X}_0$, $\pi^n(x)$ does not exist when $\pi^{n-1}(x) \not\in \mathcal{X}_0$.
Using $\pi^n(x)$ and $\pi^{-n}(x)$, for $n \in \mathbb{N}$, we define $\tilde{T}_{x,n} (y)$ and $\tilde{T}_{x,-n}(y)$ as follows:
\begin{eqnarray}
& \tilde{T}_{x,0}(y) \equiv y \ \ \ \ \ \ \ \ \ \ \ \ \ \ \ \ \ \ \ \ \ \ \ \ \ \ \ \ \ \ \ \ \ & \text{ for all $x \in \mathcal{X}$,} \nonumber \\
& \tilde{T}_{x,n}(y) \equiv \tilde{T}_{\pi^{n-1}(x),1} \circ \cdots \circ \tilde{T}_{x,1}(y) \ \ \ \ \ & \text{ if $\pi^{n}(x)$ exists,} \nonumber \\
& \tilde{T}_{x,-n}(y) \equiv \tilde{T}_{\pi^{-(n-1)}(x),-1} \circ \cdots \circ \tilde{T}_{x,-1}(y) & \text{ if $\pi^{-n}(x)$ exists.} \nonumber 
\end{eqnarray}
Then, if $\pi^{n}(x)$ exists for $n \in \mathbb{Z}$, we have
$$
g(\pi^n(x),e) = \tilde{T}_{x,n} \left( g(x,e) \right),
$$
$\tilde{T}_{x,n}(y)$ is strictly increasing in $y$, and $\tilde{T}_{x,n} ([\underline{y},\overline{y}]) = [\underline{y},\overline{y}]$.

This result implies that, if $\pi^n(x)$ exists for $n \in \mathbb{Z}$, we have $\tilde{T}_{x,n}(y) = T_{\pi^n(x),x}(y)$, and hence $T_{\pi^n(x),x}(y)$ is identified.
This information restricts the functional form of $g$.
However, as in Section 3.2, it does not provide the informative bounds of $g$ without other restrictions.

Here, we examine the properties of $\pi(x)$ and $\pi^{-1}(x)$.
Because $F_{X|Z}(x|0) < F_{X|Z}(x|1)$ for $x \in \mathcal{X}$, we have
\begin{equation}\label{pi_property1}
\begin{aligned}
\pi(x) &=& Q_{X|Z}\left( F_{X|Z}(x|0) | 1 \right) \ < \ Q_{X|Z}\left( F_{X|Z}(x|1) | 1 \right) \ = \ x, \\
 \pi^{-1}(x) &=& Q_{X|Z}\left( F_{X|Z}(x|1) | 0 \right) \ > \ Q_{X|Z}\left( F_{X|Z}(x|0) | 0 \right) \ = \ x.
\end{aligned}
\end{equation}
Figure 1 illustrates this intuitively.
Because $X|Z=0$ stochastically dominates $X|Z=1$ and functions $\pi(x)$ and $\pi^{-1}(x)$ satisfy (\ref{pi_property2}), the inequalities hold.

%%%%%%%%%%%%%%%%%%%%%%%%%%%%%%%%%%%%%%%%%%%%
\subsection{Review of \cite{d2015identification} and \cite{torgovitsky2015identification}}

To facilitate the illustration of our identification results, we first review the identification approach of \cite{d2015identification} and \cite{torgovitsky2015identification} when $\underline{x}_0 = \underline{x}_1 = \xi$, although Assumption 3 rules out the case of $\underline{x}_0 = \underline{x}_1 = \xi$.
Additionally, we assume that $g(\xi,e)$ exists and is strictly increasing in $e$.

\cite{d2015identification} and \cite{torgovitsky2015identification} use function $T_{x',x}(y)$ to identify the structural function $g$.
By definition, this function satisfies $g(x',e) = T_{x',x} \left( g(x,e) \right)$.
This function corresponds to $Q_{x'x}$ in \cite{d2015identification}.
We define 
$$
G_x(u) \equiv \int F_{Y|X=x'} \left( T_{x',x}(u) \right) dF_{X}(x').
$$
Then, similar to (\ref{approach_DF1}), we have $G_x \left(g(x,e)\right) = e$, and hence $g(x,e) =  G_x^{-1}(e)$.
If we can identify $T_{x',x}(y)$ for all $x$ and $x'$, we then can point identify the structural function $g$.

Pick an initial point $x_0 \in \mathcal{X}$ (i.e., $x_0 > \xi$) and form a recursive sequence $x_{n+1} = \pi(x_n)$ for $n > 0$.
Because $\underline{x}_0 = \underline{x}_1 = \xi$ implies $\mathcal{X}_1 \subset \mathcal{X}_0$, we have $\pi(x) \in \mathcal{X}_0$ for all $x \in \mathcal{X}$ and there exists a sequence $\{\pi^n(x)\}_{n=1}^{\infty}$.
The sequence $\{x_n\}$ is decreasing by (\ref{pi_property1}) and $x_n > \xi$ for all $n \geq 0$ by the definition of $\pi(x)$.
Hence, sequence $\{x_n\}$ converges to a limiting point.
Because (\ref{pi_property2}) implies 
$$
 F_{X|Z}(x_{n+1}|1) = F_{X|Z}(x_n|0)
$$
and $F_{X|Z}(x|z)$ is continuous in $x$, we have $F_{X|Z}(\lim_{n \rightarrow \infty}x_n|1) = F_{X|Z}(\lim_{n \rightarrow \infty}x_n|0)$.
Because $F_{X|Z}(x|0) < F_{X|Z}(x|1)$ for all $x \in (\xi, \overline{x}_0)$ and $F_{X|Z}(\xi|0) = F_{X|Z}(\xi|1) = 0$, the sequence $\{x_n\}$ converges to $\xi$ for any initial point $x_0 \in \mathcal{X}$.
Figure 2 illustrates this intuitively.
Then, for all $x \in \mathcal{X}$ and $e \in (0,1)$, we obtain
$$
\lim_{n \to \infty} \tilde{T}_{x,n}\left( g(x,e) \right) = \lim_{n \to \infty} g(\pi^n(x),e) = g(\xi,e).
$$
By substituting $g^{-1}(x,y)$ for $e$, we have $\lim_{n \to \infty} \tilde{T}_{x,n}(y)=g\left(\xi,g^{-1}(x,y)\right)$.
Hence, $T_{\xi, x}(y)$ is identified for all $x \in \mathcal{X}$.
By definition of $T_{x',x}(y)$, we have
$$
T_{x',x}(y) = T_{\xi,x'}^{-1}\left( T_{\xi,x}(y) \right).
$$
This implies that $T_{x',x}(y)$ is identified for all $x$ and $x'$.
Hence, as previously discussed, $g$ is point identified.

This approach is not available under Assumption 3 because a convergent sequence $\{\pi^n(x)\}_{n=1}^{\infty}$ does not exist.
When $F_{X|Z}(x|0)$ and $F_{X|Z}(x|1)$ have no intersections, $\pi^n(x)$ lies in $\mathcal{X}_1 \cap \mathcal{X}_0^c = (\underline{x}_1,\underline{x}_0]$ when $n$ is sufficiently large.
If $\pi^n(x)$ is in $\mathcal{X}_1 \cap \mathcal{X}_0^c$, then $\pi^{n+1}(x)$ does not exist.
From the proof of Lemma 1, for all $x \in \mathcal{X}$, $\{n: \pi^n(x)\ \text{exists.} \}$ is a finite set under Assumption 3.
For example, in Figure 1, $\pi(x)$, $\pi^{-1}(x)$, and $\pi^{-2}(x)$ exist but $\pi^2(x)$ and $\pi^{-3}(x)$ do not.

%%%%%%%%%%%%%%%%%%%%%%%%%%%%%%%%%%%%%%%%%%%%
\subsection{Unidentifiability under no shape restrictions}

In this section, we show that if we do not impose additional restrictions beyond Assumptions 1--5, the identified set of $g(x,e)$ can become unbounded.
To show this, we derive the identified set of $g$. 
We define
\begin{eqnarray}
\mathcal{G} & \equiv & \left\{ \tilde{g}:\mathcal{X}\times (0,1) \rightarrow \mathbb{R}: \text{$\tilde{g}(x,e)$ is continuous and strictly increasing in $e$.} \right\}. \nonumber
\end{eqnarray}
\cite{torgovitsky2015identification} derives the identified set of $g$ under another normalization assumption.
Similar to \cite{torgovitsky2015identification}, we obtain the following identified set:
$$
\mathcal{G}_I \equiv \left\{ \tilde{g} \in \mathcal{G} : (\tilde{g}^{-1}(X,Y),V) \indepe Z \ \ \text{and} \ \ \tilde{g}^{-1}(X,Y) \sim U(0,1) \right\},
$$
where $\tilde{g}^{-1}$ is the inverse of $\tilde{g}$ with respect to its last component and $V$ is defined as in (\ref{def_V}).
The independence condition in the identified set is equivalent to the following condition:
$$
P\left( Y \leq \tilde{g}(X,e) |V=v, Z=0 \right) = P\left( Y \leq \tilde{g}(X,e) |V=v, Z=1 \right) \ \ \text{for all $v\in (0,1)$.}
$$
From the definition of $V$, for all $v \in (0,1)$, we have
$$
F_{Y|X,Z}\left( \tilde{g}(x_{v,0},e)|x_{v,0},0 \right) \ = \ F_{Y|X,Z}\left( \tilde{g}(x_{v,1},e)|x_{v,1},1 \right),
$$
where $x_{v,z} \equiv Q_{X|Z}(v|z)$.
Hence, we can rewrite $\mathcal{G}_I$ as
\begin{eqnarray}
\mathcal{G}_I &=& \left\{ \tilde{g} \in \mathcal{G}: \tilde{g}^{-1}(X,Y) \sim U(0,1) \ \ \text{and}  \right. \nonumber \\
& & \left. \tilde{g}\left(x_{v,1},\tilde{g}^{-1}(x_{v,0},\cdot)\right) = Q_{Y|X,Z} \left(F_{Y|X,Z}(\cdot|x_{v,0},0)|x_{v,1},1 \right) \ \text{for all $v$.} \right\}. \label{Identified_set_of_g}
\end{eqnarray}
This expression implies that $T_{x_{v,1}, x_{v,0}}(y)$ is identified for all $v$.
Proposition 1 provides the same result.
The sharp lower and upper bounds of $g(x,e)$ are obtained by $\inf_{\tilde{g}\in \mathcal{G}_I} \tilde{g}(x,e)$ and $\sup_{\tilde{g}\in \mathcal{G}_I} \tilde{g}(x,e)$.

To show that the bounds of $g(x,e)$ can be unbounded, we consider the following simple model:
\begin{eqnarray}
Y &=& \Phi^{-1}(\epsilon), \nonumber \\
X &=& Z(\eta-1) + (1-Z)\eta, \nonumber
\end{eqnarray}
where $\Phi(\cdot)$ is the standard normal distribution function, $\epsilon \sim U(0,1)$, $\eta \sim U(0,1)$, $Z$ is a random Bernoulli variable with $p=0.5$, and $(\epsilon,\eta,Z)$ are mutually independent.
Then, it follows from (\ref{Identified_set_of_g}) that $\tilde{g} \in \mathcal{G}_I$ if and only if
\begin{eqnarray}
& & \tilde{g}(v,e)=\tilde{g}(v-1,e) \ \ \text{for all $e,v \in (0,1)$,} \label{tilde_g_1} \\
& & P(Y \leq \tilde{g}(X,e)) = e \ \ \text{for all $e \in (0,1)$.} \label{tilde_g_2}
\end{eqnarray}
We construct $\tilde{g}_K$ as follows.
First, we define
\begin{equation}
\tilde{g}_K(x,0.5) \equiv \begin{cases}
    \Phi^{-1}\left( 4K(x+1)+0.5-K \right), & -1<x\leq -0.5 \\
    \Phi^{-1}\left( -4K(x+0.5)+0.5+K \right), & -0.5<x\leq 0 \\
    \Phi^{-1}\left( 4Kx+0.5-K \right), & 0<x\leq 0.5 \\
    \Phi^{-1}\left( -4K(x-0.5)+0.5+K \right), & 0.5<x\leq 1
  \end{cases}, \nonumber
\end{equation}
where $-0.5 < K < 0.5$.
Second, for $e \neq 0.5$, we define $\tilde{g}_K(x,e)$ as
\begin{eqnarray}
\tilde{g}_K(x,e) \equiv \begin{cases}
    \Phi^{-1}\left( 2e \Phi(\tilde{g}_K(x,0.5)) \right), & 0<e<0.5 \\
    \Phi^{-1}\left( 1-2(1-e)\{1-\tilde{g}_K(x,0.5)\} \right), & 0.5<e<1 
  \end{cases}. \nonumber
\end{eqnarray}
Then, we confirm that $\tilde{g}_K$ satisfies (\ref{tilde_g_1}) and (\ref{tilde_g_2}) for all $-0.5 < K < 0.5$.
Hence, $\tilde{g}_K$ is an element of $\mathcal{G}_I$ for all $-0.5 < K < 0.5$.
Because $\tilde{g}_K(0,0.5)=\Phi^{-1}(0.5-K)$, the lower and upper bounds of $g(0,0.5)$ are $-\infty$ and $+ \infty$, respectively.
Therefore, in this setting, the identified set of $g$ can be unbounded.

%%%%%%%%%%%%%%%%%%%%%%%%%%%%%%%%%%%%%%%%%%%
\section{Bounds under additional shape restrictions}

We show the partial identification of $g$ under some shape restrictions.
In Sections 4.1, 4.2, and 4.3, we show the partial identification under monotonicity, concavity, and monotonicity and concavity, respectively.
In Section 4.4, we show that point identification can be achieved when the structural function is flat or linear with respect to $x$ over a given interval.

%%%%%%%%%%%%%%%%%%%%%%%%%%%%%%%%%%%%%%%%%%%
\subsection{Bounds under monotonicity}

In this section, we propose a method to construct the lower and upper bounds of $g(x,e)$ under monotonicity.
First, we show that a set $\Pi_{x',x}^M$ defined below is nonempty and finite, when $F_{X|Z}(x|0)$ and $F_{X|Z}(x|1)$ have no intersections.
Second, we show that we can partially identify $g(x,e)$ using $\Pi_{x',x}^M$ when $g(x,e)$ is nondecreasing in $x$.

For $(x,x') \in \mathcal{X} \times \mathcal{X}$, we define $\Pi_{x',x}^M$ as
\begin{eqnarray}
 \Pi_{x',x}^M &\equiv & \left\{ (n,m) : n,m \in \mathbb{Z}, \ \text{$\pi^n(x')$ and $\pi^m(x)$ exist, and $\pi^n(x') \leq \pi^m(x)$.} \right\}. \label{Pi}
\end{eqnarray}
In Figure 1, $\Pi^{M}_{x',x}=\{(-1,-2),(0,-2), (0,-1), (1,-2), (1,-1), (1,0) \}$.
The following lemma shows that $\Pi_{x',x}^M$ is nonempty and finite when $F_{X|Z}(x|0)$ and $F_{X|Z}(x|1)$ have no intersections.

%%% Lemma 1 %%%
\begin{Lemma}
Under Assumptions 1--5, $\Pi_{x',x}^M$, as defined by (\ref{Pi}), is nonempty and finite for all $(x,x') \in \mathcal{X} \times \mathcal{X}$.
\end{Lemma}

Under Assumptions 1--5, for any $x \in \mathcal{X}$ the set $\{n \in \mathbb{Z}: \pi^n(x)\ \text{exists.} \}$ is finite from the proof of Lemma 1.
Hence, $g$ cannot be point identified using the method proposed by \cite{d2015identification} and \cite{torgovitsky2015identification}).

We impose the following assumption:

%%% Assumption 6 %%%
\begin{Assumption}[Monotonicity]
For $e \in (0,1)$, $g(x,e)$ is nondecreasing in $x$.
\end{Assumption}

The monotonicity assumption holds for many economic models.
For example, the demand function is ordinarily decreasing in price if the income effect is negligible, and economic analyses of production often assume that the production function is monotonically increasing in input.
Monotonicity assumptions of this type have been employed in many studies.
For example, \cite{manski1997monotone} imposes a monotonicity assumption on a response function and shows that the average treatment response is partially identified.

If $(n,m) \in \Pi_{x',x}^M$, Assumption 6 implies that
$$
 \tilde{T}_{x',n} \left( g(x',e) \right) = g(\pi^n(x'),e) \leq g(\pi^m(x),e) = \tilde{T}_{x,m} \left( g(x,e) \right).
$$
Because $\tilde{T}_{x',n}(y)$ is strictly increasing in $y$ and $\tilde{T}_{x',n}\left( [\underline{y},\overline{y}] \right) = [\underline{y},\overline{y}]$, we have $g(x',e) \leq \tilde{T}_{x',n}^{-1} \left( \tilde{T}_{x,m} \left( g(x,e) \right) \right)$ for $(n,m)\in \Pi_{x',x}^M$.
Hence, we have
\begin{equation}
 g(x',e) \leq \min_{(n,m)\in \Pi_{x',x}^M} \tilde{T}_{x',n}^{-1} \left( \tilde{T}_{x,m} \left( g(x,e) \right) \right). \nonumber
\end{equation}
Define
\begin{equation}
\label{def_TMU_TML}
\begin{aligned}
T_{x',x}^{MU}(y) & \equiv \min_{(n,m)\in \Pi_{x',x}^M} \tilde{T}_{x',n}^{-1} \left( \tilde{T}_{x,m} \left( y \right) \right), \\
T_{x',x}^{ML}(y) & \equiv \max_{(n,m)\in \Pi_{x,x'}^M}  \tilde{T}_{x',m}^{-1} \left( \tilde{T}_{x,n} \left( y \right) \right).
\end{aligned}
\end{equation}
Then, $T_{x',x}^{MU}(y)$ is strictly increasing and satisfies
\begin{equation}
g(x',e) \leq T_{x',x}^{MU}\left( g(x,e) \right). \label{TU_M} 
\end{equation}
Similarly, $T_{x',x}^{ML}(y)$ is strictly increasing and satisfies
\begin{equation}
 g(x',e) \geq T_{x',x}^{ML}\left( g(x,e) \right). \label{TL_M}
\end{equation}

As already mentioned, the functions that satisfy (\ref{TU}) and (\ref{TL}) are the upper and lower bounds of $T_{x',x}(y)$, respectively.
Hence, for any $(n,m) \in \Pi_{x',x}^M$, $\tilde{T}_{x',n}^{-1} \left( \tilde{T}_{x,m}(y) \right)$ becomes an upper bound of $T_{x',x}(y)$.
This implies that $T^{MU}_{x',x}(y)$ is the lowest upper bound of $T_{x',x}(y)$ in the sense that $T^{MU}_{x',x}(y)$ is lower than $\tilde{T}_{x',n}^{-1} \left( \tilde{T}_{x,m}(y) \right)$ for any $(n,m) \in \Pi_{x',x}^M$.
Similarly, $T^{ML}_{x',x}(y)$ is the largest lower bound of $T_{x',x}(y)$.

We define 
\begin{eqnarray}
G_x^{ML}(u) &\equiv & \int F_{Y|X=x'} \left( T_{x',x}^{MU} (u) \right) dF_{X}(x'), \nonumber \\
G_x^{MU}(u) &\equiv & \int F_{Y|X=x'} \left( T_{x',x}^{ML} (u) \right) dF_{X}(x') , \nonumber \\
B^{ML}(x,e) &\equiv & \sup_{y:y \leq x} \left\{ \left(G_y^{ML}\right)^{-1}(e) \right\}, \nonumber \\
B^{MU}(x,e) &\equiv & \inf_{y:y \geq x} \left\{ \left(G_y^{MU}\right)^{-1}(e) \right\}. \nonumber
\end{eqnarray}
$G_x^{ML}(u)$ and $G_x^{MU}(u)$ provide the lower and upper bounds of $g(x,e)$ on the basis of arguments (\ref{approach_DF1}) and (\ref{approach_DF2}).
$B^{ML}(x,e)$ and $B^{MU}(x,e)$ strengthen these bounds.

%%% Theorem 1 %%%
\begin{Theorem}
Under Assumptions 1--6, for all $(x,e) \in \mathcal{X} \times (0,1)$, we have
\begin{equation}
B^{ML}(x,e) \ \leq \ g(x,e) \ \leq \ B^{MU}(x,e). \nonumber
\end{equation}
\end{Theorem}

In the first step, we show that $\left(G_x^{ML}\right)^{-1}(e) \leq g(x,e) \leq \left(G_x^{MU}\right)^{-1}(e)$.
In the second step, we strengthen these bounds to $B^{ML}(x,e) \leq g(x,e) \leq  B^{MU}(x,e)$.
Figure 3 intuitively illustrates this proof.
The idea is similar to that of \cite{manski1997monotone}, who considers the case in which response function $y(t)$ is increasing, where $y(t)$ is a latent outcome with treatment $t$.
He then uses the monotonicity of $y(t)$ to partially identify average response function $E[y(t)]$ when the support of the outcome is bounded.
By contrast, our bounds are bounded even when the support of the outcome is unbounded.

%%%%%%%%%%%%%%%%%%%%%%%%%%%%%%%%%%%%%%%%%%%%
\begin{Simulation}
To illustrate Theorem 1, we consider the following example:
\begin{eqnarray}
\label{Simulation}
\begin{aligned}
Y &= h(X) exp \left(\alpha + \beta \Phi^{-1}(\epsilon) \right) \\
X &= (0.2 + \eta)Z + (1-Z)\{(2-\rho)(\eta-1)+2.2\} ,
\end{aligned}
\end{eqnarray}
where $h(x)$ is an increasing function specified below, $\Phi(\cdot)$ is the standard normal distribution function, $Z$ is a random Bernoulli variable with $p=0.5$, and $(\alpha,\beta)=(0.5,0.5)$.
Suppose that 
\begin{eqnarray}
\epsilon &=& \Phi(U) \nonumber \\
\eta &=& \Phi(V) \nonumber \\
(U,V) &\sim & N\left( \left( \begin{array}{c}
0 \\ 0
\end{array} \right) ,\left( \begin{array}{cc}
1 & 0.3 \\ 0.3 & 1
\end{array} \right) \right). \nonumber
\end{eqnarray}
Then, $\epsilon \sim U(0,1)$ and $\eta \sim U(0,1)$.
In this example, $F_{X|Z}(x|1) = x-0.2$ for $x \in [0.2,1]$ and $F_{X|Z}(x|0) = \frac{1}{2-\rho}(x-2.2)+1$ for $x \in [\rho+0.2,2.2]$.
These functions are depicted in Figure 4.
Conditional distribution functions $F_{X|Z}(x|0)$ and $F_{X|Z}(x|1)$ do not intersect when $\rho > 0$.
When $\rho = 0$, these functions intersect at $x=0.2$.
\cite{torgovitsky2015identification} shows  that $g$ is point identified when $\rho = 0$.

We calculate the bounds of $g(x,0.5)$ using Theorem 1 when $h(x)= h_1(x) \equiv x \ \text{or} \ h(x)= h_2(x) \equiv 2\exp(4(x-1.2))/\{1+\exp(4(x-1.2))\} + 0.2$.
Figures 5 and 6 show these bounds for three different choices of $\rho$: $0.01$, $0.1$, and $0.3$.
For $h_1$ and $h_2$, the bounds become tighter as $\rho$ become smaller.
In particular, the bounds are very close to the true function when $\rho = 0.01$.
This confirms our theoretical result that $B^{ML}(x,e)$ and $B^{MU}(x,e)$ converge to $g(x,e)$ as $\rho \rightarrow 0$.
When $\rho = 0.01$ and $0.1$, the bounds of $h_2$ are tighter than that of $h_1$.
This result is caused by $h_2(x)$ being flatter than $h_1(x)$ over a particular interval.
As discussed later, Theorem 3 shows that $g$ is point identified when $g(x,e)$ is flat with respect to $x$ over a given interval.
\end{Simulation}
%%%%%%%%%%%%%%%%%%%%%%%%%%%%%%%%%%%%%%%%%%%%

%%%%%%%%%%%%%%%%%%%%%%%%%%%%%%%%%%%%%%%%%%%%
\begin{Remark}
Although our bounds may not be sharp in general, we can derive the identified set of $g$ under Assumption 6.
We define
$$
\mathcal{G}^M \equiv \{\tilde{g} \in \mathcal{G} : \text{$\tilde{g}(x,e)$ is nondecreasing in $x$.} \}.
$$
Then, similar to (\ref{Identified_set_of_g}), the identified set of $g$ under Assumption 6 is obtained by
\begin{eqnarray}
\mathcal{G}_I^M &=& \left\{ \tilde{g} \in \mathcal{G}^M : \tilde{g}^{-1}(X,Y) \sim U(0,1) \ \ \text{and}  \right. \nonumber \\
& & \left. \tilde{g}\left(x_{v,1},\tilde{g}^{-1}(x_{v,0},\cdot)\right) = Q_{Y|X,Z} \left(F_{Y|X,Z}(\cdot|x_{v,0},0)|x_{v,1},1 \right) \ \text{for all $v$.} \right\}. \nonumber
\end{eqnarray}
Hence, the sharp lower and upper bounds of $g(x,e)$ are $\inf_{\tilde{g} \in \mathcal{G}_I^M} \tilde{g}(x,e)$ and $\sup_{\tilde{g} \in \mathcal{G}_I^M} \tilde{g}(x,e)$, respectively.
However, these bounds may not coincide with $B^{ML}(x,e)$ and $B^{MU}(x,e)$.
Actually, in some settings, $\left(G_x^{ML}\right)^{-1}(e)$ and $\left(G_x^{MU}\right)^{-1}(e)$ are not nondecreasing in $x$.
This implies that $\left(G_x^{ML}\right)^{-1}(e)$ and $\left(G_x^{MU}\right)^{-1}(e)$ are not sharp in general.

It is difficult to compute $\mathcal{G}_I^M$ because $\mathcal{G}^M$ is infinite dimensional.
By contrast, $B^{ML}(x,e)$ and $B^{MU}(x,e)$ have closed-form expressions and are hence computable.
In Simulation 1, we compute $B^{ML}(x,e)$ and $B^{MU}(x,e)$ in some settings, and in Section 5, we show that our bounds are informative in real data.
\end{Remark}

%%%%%%%%%%%%%%%%%%%%%%%%%%%%%%%%%%%%%%%%%%%
\subsection{Bounds under concavity}

In this section, we propose a method to construct the lower and upper bounds of $g(x,e)$ under concavity.
First, we show that a set $\Pi_{x',x}^C$ defined below is nonempty and finite.
Second, we show that we can partially identify $g$ using $\Pi_{x',x}^C$ when $g(x,e)$ is concave in $x$.

For $(x,x') \in \mathcal{X} \times \mathcal{X}$, we define $\Pi^C_{x',x}$ as
\begin{eqnarray}
\Pi^C_{x',x} &\equiv & \left\{ (n,m) : n,m \in \mathbb{Z}, \ \text{$\pi^n(x')$, $\pi^{n-1}(x)$ and $\pi^m(x)$ exist,}\right. \nonumber \\
 & & \left. \ \ \ \ \ \ \ \text{and $\pi^n(x') \leq \pi^m(x) \leq \pi^{n-1}(x')$.} \right\}. \label{Pi^C}
\end{eqnarray}
In Figure 1, $\Pi^C_{x',x} = \{ (0,-1), (1,0) \}$.
The following lemma shows that $\Pi^C_{x',x}$ is nonempty and finite, similar to Lemma 1.

%%% Lemma 2 %%%
\begin{Lemma}
Under Assumptions 1--5, $\Pi^C_{x',x}$ as defined by (\ref{Pi^C}) is nonempty and finite for all $(x,x') \in \mathcal{X} \times \mathcal{X}$.
\end{Lemma}

Similar to Section 3, we impose the following assumption.

%%% Assumption 7 %%%
\begin{Assumption}[Concavity]
For $e \in (0,1)$, $g(x,e)$ is concave in $x$.
\end{Assumption}

The concavity assumption holds in many economic models.
For example, economic analyses of production often assume that the production function is concave in inputs.
For instance, \cite{manski1997monotone} assumes concavity and shows that the average treatment response is partially identified.
Further, \cite{d2013nonlinear} achieves the partial identification of the average treatment on the treated effect using a locally concavity assumption.

As in Section 3, if we identify functions $T_{x',x}^U(y)$ and $T_{x',x}^L(y)$ that are strictly increasing in $y$, surjective, and satisfy (\ref{TU}) and (\ref{TL}), we can obtain the lower and upper bounds of $g(x,e)$.
Hence, we consider constructing functions $T_{x',x}^U(y)$ and $T_{x',x}^L(y)$ that are strictly increasing in $y$, surjective, and satisfy (\ref{TU}) and (\ref{TL}).

If $(n,m) \in \Pi^C_{x',x}$, from Assumption 7, we have
$$
t_{x',x}^{n,m} \cdot \tilde{T}_{x',n}\left( g(x',e) \right) + (1-t_{x',x}^{n,m}) \cdot \tilde{T}_{x',n-1}\left( g(x',e) \right) \leq \tilde{T}_{x,m} \left( g(x,e) \right),
$$
where $t_{x',x}^{n,m}= \left( \pi^{n-1}(x') - \pi^m(x) \right) / \left( \pi^{n-1}(x') - \pi^n(x') \right)$.
We define
$$
\tilde{\mathcal{T}}_{x',x,n,m}(y) \equiv t_{x',x}^{n,m} \cdot \tilde{T}_{x',n}(y) + (1-t_{x',x}^{n,m}) \cdot \tilde{T}_{x',n-1}(y).
$$
Because $\tilde{T}_{x',n}(y)$ and $\tilde{T}_{x',n-1}(y)$ are surjective and strictly increasing in $y$, we obtain 
$$
g(x',e) \leq \min_{(n,m)\in\Pi^C_{x',x}} \tilde{\mathcal{T}}_{x',x,n,m}^{-1} \left( \tilde{T}_{x,m} \left( g(x,e) \right) \right).
$$
Define
\begin{equation}\label{def_TCU_TCL}
\begin{aligned}
T_{x',x}^{CU}(y) & \equiv & \min_{(n,m)\in\Pi^C_{x',x}} \tilde{\mathcal{T}}_{x',x,n,m}^{-1} \left( \tilde{T}_{x,m} \left( y \right) \right), \\
T_{x',x}^{CL}(y) & \equiv & \max_{(n,m)\in\Pi^C_{x,x'}} \tilde{T}_{x',m}^{-1} \left( \tilde{\mathcal{T}}_{x,x',n,m}(y) \right).
\end{aligned}
\end{equation}
Then, $T_{x',x}^{CU}(y)$ and $T_{x',x}^{CL}(y)$, as defined in (\ref{def_TCU_TCL}), are strictly increasing in $y$ and satisfy
\begin{eqnarray}
 g(x',e) \leq T_{x',x}^{CU} \left( g(x,e) \right), \label{TU_C} \\
 g(x',e) \geq  T_{x',x}^{CL} \left( g(x,e) \right). \label{TL_C} 
\end{eqnarray}

We define
\begin{eqnarray}
G_x^{CL}(u) &\equiv & \int F_{Y|X=x'} \left( T_{x',x}^{CU} (u) \right) dF_{X}(x'), \nonumber \\
G_x^{CU}(u) &\equiv & \int F_{Y|X=x'} \left( T_{x',x}^{CL} (u) \right) dF_{X}(x'), \nonumber
\end{eqnarray}
$$
 B^{CL}(x,e) \equiv \sup_{y,y':y < x < y'} \left\{ \left( \frac{x-y}{y'-y} \right) \left( G_{y'}^{CL} \right)^{-1}(e) + \left( \frac{y'-x}{y'-y} \right) \left( G_y^{CL} \right)^{-1}(e) \right\},
$$
\begin{eqnarray}
 B^{CU}(x,e) \equiv & \min \left[  \inf_{y,y':x < y < y'} \left\{ \left( \frac{x-y}{y'-y} \right)B^{CL}(y',e) + \left( \frac{y'-x}{y'-y} \right) \left( G_y^{CU} \right)^{-1}(e) \right\}, \right. \nonumber \\
 & \left. \inf_{y,y':y' < y < x} \left\{ \left( \frac{y-x}{y-y'} \right)B^{CL}(y',e) + \left( \frac{x-y'}{y-y'} \right) \left( G_y^{CU} \right)^{-1}(e) \right\} \right]. \nonumber
\end{eqnarray}
$G_x^{CL}(u)$ and $G_x^{CU}(u)$ provide the lower and upper bounds of $g(x,e)$ as per (\ref{approach_DF1}) and (\ref{approach_DF2}).
$B^{CL}(x,e)$ and $B^{CU}(x,e)$ strengthen these bounds.

%%% Theorem 2 %%%
\begin{Theorem}
Under Assumptions 1--5 and 7, for all $(x,e) \in \mathcal{X} \times (0,1)$, we have
$$
 B^{CL}(x,e) \leq g(x,e) \leq B^{CU}(x,e).
$$
\end{Theorem}

Similar to Theorem 1, we can show that $\left( G_x^{CL} \right)^{-1}(e) \leq g(x,e) \leq \left( G_x^{CU} \right)^{-1}(e)$.
We strengthen the bounds to $B^{CL}(x,e) \leq g(x,e) \leq B^{CU}(x,e)$ using the concavity of $g(x,e)$ in $x$.
Figure 7 intuitively illustrates this proof.
A similar approach is used by \cite{manski1997monotone}, namely utilizing the concavity of the response function to partially identify the average response function when the support of the outcome is bounded.
However, our approach does not require information on the infimum and supremum of the support of the outcome.

This identification approach is somewhat similar to that of \cite{d2013nonlinear}, who study the identification of nonseparable models with continuous, endogenous regressors, using repeated cross sections.
Specifically, they consider the following model:
$$
Y_t = g_t(X_t,A_t), \ \ t = 1, \cdots, T,
$$
where $A_t$ is an unobserved heterogeneous factor.
They show that, under the assumptions that $A_t|V_t \equiv F_{X_t}(X_t)=v \sim A_s|V_s \equiv F_{X_s}(X_s)=v$ and $g_t(x,a) = m_t(g(x,a))$, the average treatment on treated effect $\Delta^{ATT}(x,x') \equiv \mathbb{E}[g_T(x,A_T)-g_T(x',A_T)|X_T=x]$ is identified when $F_{X_T}(x) = F_{X_t}(x')$.
Under this assumption, $\Delta^{ATT}(x,x')$ is not identified if $F_{X_T}(x) \neq F_{X_t}(x')$ for all $t \in \{1, \cdots , T-1\}$.
However, they show that $\Delta^{ATT}(x,x')$ is partially identified if $x \mapsto g(x,a)$ is locally concave.

%%%%%%%%%%%%%%%%%%%%%%%%%%%%%%%%%%%%%%%%%%%
\begin{Simulation}
To illustrate Theorem 2, we consider model (\ref{Simulation}).
We set $h(x) = -(x-1.2)^2 + 1.5$ and calculate the bounds of $g(x,0.5)$ using Theorem 2.
Figure 8 shows $B^{CL}(x,0.5)$ and $B^{CU}(x,0.5)$ for three different choices of $\rho$: 0.1, 0.5, and 0.8.
Similar to Simulation 1, the bounds become tighter as $\rho$ become smaller.
In particular, the bounds are very close to the true function when $\rho = 0.1$.
This confirms our theoretical result that $B^{CL}(x,e)$ and $B^{CU}(x,e)$ converge to $g(x,e)$ as $\rho \rightarrow 0$.
\end{Simulation}
%%%%%%%%%%%%%%%%%%%%%%%%%%%%%%%%%%%%%%%%%%%

%%%%%%%%%%%%%%%%%%%%%%%%%%%%%%%%%%%%%%%%%%%
\subsection{Bounds under monotonicity and concavity}

In several cases, such as the production function, we can assume that both Assumptions 6 and 7 hold.
Then, it follows from Theorems 1 and 2 that
\begin{equation}
\max \{B^{ML}(x,e),B^{CL}(x,e)\} \leq g(x,e) \leq \min \{B^{MU}(x,e),B^{CU}(x,e)\}. \label{bounds_M_C_1}
\end{equation}
In this case, we can obtain tighter bounds in the following manner.
We define
\begin{eqnarray}
T_{x',x}^{MCU}(y) & \equiv & \min\{T^{MU}_{x',x}(y),T^{CU}_{x',x}(y)\}, \nonumber \\
T_{x',x}^{MCL}(y) & \equiv & \max\{T^{ML}_{x',x}(y),T^{CL}_{x',x}(y)\}, \nonumber \\ 
G_x^{MCL}(u) &\equiv & \int F_{Y|X=x'} \left( T_{x',x}^{MCU} (u) \right) dF_{X}(x'), \nonumber \\
G_x^{MCU}(u) &\equiv & \int F_{Y|X=x'} \left( T_{x',x}^{MCL} (u) \right) dF_{X}(x'). \nonumber
\end{eqnarray}
Similarly to the above arguments, we have $g(x',e) \leq T_{x',x}^{MCU}(g(x,e))$ and $g(x',e) \geq T_{x',x}^{MCL}(g(x,e))$, and hence we can obtain
\begin{eqnarray}
\left( G_x^{MCL} \right)^{-1}(e) \leq g(x,e) \leq \left( G_x^{MCU} \right)^{-1}(e). \nonumber
\end{eqnarray}
We define
\begin{eqnarray}
\tilde{B}^{MCL}(x,e) &\equiv & \sup_{y:y \leq x} \left\{ \left(G_y^{MCL}\right)^{-1}(e) \right\}, \nonumber \\
\tilde{B}^{MCU}(x,e) &\equiv & \inf_{y:y \geq x} \left\{ \left(G_y^{MCU}\right)^{-1}(e) \right\}, \nonumber 
\end{eqnarray}
$$
 \hat{B}^{MCL}(x,e) \equiv \sup_{y,y':y < x < y'} \left\{ \left( \frac{x-y}{y'-y} \right) \left( G_{y'}^{MCL} \right)^{-1}(e) + \left( \frac{y'-x}{y'-y} \right) \left( G_y^{MCL} \right)^{-1}(e) \right\},
$$
\begin{eqnarray}
 \hat{B}^{MCU}(x,e) \equiv & \min \left[  \inf_{y,y':x < y < y'} \left\{ \left( \frac{x-y}{y'-y} \right)\hat{B}^{MCL}(y',e) + \left( \frac{y'-x}{y'-y} \right) \left( G_y^{MCU} \right)^{-1}(e) \right\}, \right. \nonumber \\
 & \left. \inf_{y,y':y' < y < x} \left\{ \left( \frac{y-x}{y-y'} \right)\hat{B}^{MCL}(y',e) + \left( \frac{x-y'}{y-y'} \right) \left( G_y^{MCU} \right)^{-1}(e) \right\} \right]. \nonumber
\end{eqnarray}
Then, from the above results, both $\tilde{B}^{MCU}(x,e)$ and $\hat{B}^{MCU}(x,e)$ are upper bounds of $g(x,e)$.
Similarly, both $\tilde{B}^{MCL}(x,e)$ and $\hat{B}^{MCL}(x,e)$ are also lower bounds of $g(x,e)$.
Therefore, we can obtain
\begin{equation}
\max\{\tilde{B}^{MCL}(x,e), \hat{B}^{MCL}(x,e)\} \leq g(x,e) \leq \min\{\tilde{B}^{MCL}(x,e),\hat{B}^{MCL}(x,e)\}. \label{bounds_M_C_2}
\end{equation}
Clearly, these bounds are tighter than (\ref{bounds_M_C_1}).

%%%%%%%%%%%%%%%%%%%%%%%%%%%%%%%%%%%%%%%%%%%
\subsection{Point identification}

First, we show that $g(x,e)$ is point identified under monotonicity when the structural function is flat in $x$ over a given interval.
The argument in Section 4.1 shows that the bounds become tighter as the difference between $g(x',e)$ and $T_{x',x}^{MU} \left( g(x,e) \right)$ (or $T_{x',x}^{ML} \left( g(x,e) \right)$) decreases.
The following theorem shows that, if $g(x,e)$ is flat in $x$ over a given interval, inequalities (\ref{TU}) and (\ref{TL}) become equalities and structural function $g$ is point identified.

%%% Theorem 3 %%%
\begin{Theorem}
Under Assumptions 1--6, if there exists $\tilde{x} \in \mathcal{X}_0 \cap \mathcal{X}_1$ such that $x \mapsto g(x,e)$ is constant on $[\pi(\tilde{x}),\pi^{-1}(\tilde{x})]$ for each $e \in (0,1)$, then $B^{ML}(x,e)$ and $B^{MU}(x,e)$ coincide with $g(x,e)$ for all $(x,e) \in \mathcal{X} \times (0,1)$.
Hence, $g$ is point identified.
This result holds even when the interval $[\pi(\tilde{x}),\pi^{-1}(\tilde{x})]$ is unknown.
\end{Theorem}

In the first step, we show that, for all $x \in \mathcal{X}$, $n \in \mathbb{Z}$ exists such that $\pi^n(x), \pi^{n+1}(x) \in [\pi(\tilde{x}),\pi^{-1}(\tilde{x})]$.
In the second step, we show $g$ is point identified.
Because $g(x,e)$ is constant in $x$ conditional on $[\pi(\tilde{x}),\pi^{-1}(\tilde{x})]$, we have $g(x',e)= T_{x',x}^{MU} \left( g(x,e) \right)$ and $g(x',e)= T_{x',x}^{ML} \left( g(x,e) \right)$ for all $x, x' \in \mathcal{X}$ and $e \in (0,1)$.
Hence, $B^{ML}(x,e)$ and $B^{MU}(x,e)$ coincide with $g(x,e)$ because inequalities (\ref{TU_M}) and (\ref{TL_M}) become equalities.

%%%%%%%%%%%%%%%%%%%%%%%%%%%%%%%%%%%%%%%%%%%%
\begin{Simulation}
To illustrate Theorem 3, we consider model (\ref{Simulation}).
We set $h(x) = \max\{0,x-\delta\}+0.5$ and $\rho=0.3$.
Figures 9--11 show $B^{ML}(x,0.5)$ and $B^{MU}(x,0.5)$ for three different choices of $\delta$: $0.4$, $0.55$, and $1.2$.
In this model, $g(x,e)$ is constant on $[0.2,\delta]$.
Because $\pi(0.5)=0.2$ and $\pi^{-1}(0.5)=1.01$, interval $[0.2,\delta]$ covers $[\pi(0.5),0.5]$ when $\delta=0.55$ and covers $[\pi(0.5),\pi^{-1}(0.5)]$ when $\delta=1.2$.
Hence, the condition of Theorem 3 is satisfied only when $\delta=1.2$.
In Figure 11, $B^{ML}(x,0.5)$ and $B^{MU}(x,0.5)$ coincide with $g(x,0.5)$ when $\delta = 1.2$.
By contrast, when $\delta = 0.4$ and $0.55$, $g(x,0.5)$ is not point identified.
\end{Simulation}
%%%%%%%%%%%%%%%%%%%%%%%%%%%%%%%%%%%%%%%%%%%%

Next, we show that $g(x,e)$ is point identified under concavity when the structural function is linear in $x$ over a given interval.
Similar to Theorem 3, the following theorem shows that, if $g(x,e)$ is linear in $x$ over a particular interval, inequalities (\ref{TU_C}) and (\ref{TL_C}) become equalities, and $B^{CL}(x,e)$ and $B^{CU}(x,e)$ coincide with $g(x,e)$.

%%% Theorem 4 %%%
\begin{Theorem}
Under Assumptions 1--5 and 7, if $\tilde{x} \in \mathcal{X}$ exists such that $g(x,e)$ is linear in $x$ on $[\pi(\tilde{x}),\pi^{-1}(\tilde{x})]$, then $B^{CL}(x,e)$ and $B^{CU}(x,e)$ coincide with $g(x,e)$.
Hence, $g$ is point-identified.
This result holds even if interval $[\pi(\tilde{x}),\pi^{-1}(\tilde{x})]$ is unknown.
\end{Theorem}

%%%%%%%%%%%%%%%%%%%%%%%%%%%%%%%%%%%%%%%%%%%%
\begin{Example}[Quantile regression models]
We assume $g(X,\epsilon)=\theta_0(\epsilon)+\theta_1(\epsilon)X$, where $\theta_0(e)+\theta_1(e)x$ is strictly increasing in $e$ for all $x \in \mathcal{X}$.
This model is a quantile regression model with endogeneity.
The $\tau$-th quantile function of $g(x,\epsilon)$ is $\theta_0(\tau)+\theta_1(\tau)x$.
In this case, structural function $g(x,e) = \theta_0(e) + \theta_1(e)x$ is linear in $x$.
Hence, Theorem 4 shows that $\theta_0(e)$ and $\theta_1(e)$ are identified if binary instruments are available.

In this case, we can identify $\theta_0\left( Q_{\epsilon|\eta}(\tau|v) \right)]$ and $\theta_1\left( Q_{\epsilon|\eta}(\tau|v) \right)$ by another approach.
As in Section 3, we obtain $\epsilon|X=Q_{X|Z}(v|z), Z=z \ \sim \ \epsilon|\eta=v$ for all $v \in (0,1)$ and $z \in \{0,1\}$.
This implies that
\begin{eqnarray}
Q_{Y|X,Z}\left(\tau|Q_{X|Z}(v|0),0 \right) &=& \theta_0\left( Q_{\epsilon|\eta}(\tau|v) \right) + \theta_1\left( Q_{\epsilon|\eta}(\tau|v) \right) \times Q_{X|Z}(v|0), \nonumber \\
Q_{Y|X,Z}\left(\tau|Q_{X|Z}(v|1),1 \right) &=& \theta_0\left( Q_{\epsilon|\eta}(\tau|v) \right) + \theta_1\left( Q_{\epsilon|\eta}(\tau|v) \right) \times Q_{X|Z}(v|1). \nonumber 
\end{eqnarray}
Because $Q_{X|Z}(v|0) \neq Q_{X|Z}(v|1)$ under Assumption 3, for all $\tau \in (0,1)$ and $v \in (0,1)$, we can obtain $\theta_0\left( Q_{\epsilon|\eta}(\tau|v) \right)$ and $\theta_1\left( Q_{\epsilon|\eta}(\tau|v) \right)$ from the above equations.
This result is similar to the identification results of \cite{chesher2003identification} and \cite{jun2009local}.

The above model is a special case of the linear correlated random coefficients (CRC) model.
\cite{masten2016identification} consider the linear CRC model and show that the expectations of coefficients are identified.
In this model, we can also identify the expectations of coefficients as $\mathbb{E}[\theta_j(\epsilon)]$.
Let $U$ be a uniformly distributed random variable.
Then, it follows from $Q_{\epsilon|\eta}(U|v) \sim \epsilon|\eta=v$ that $\int_0^1 \theta_j\left( Q_{\epsilon|\eta}(\tau|v) \right) d\tau = \mathbb{E}\left[ \theta_j\left( Q_{\epsilon|\eta}(U|v) \right) \right] = \mathbb{E}[\theta_j(\epsilon)|\eta=v]$.
Hence, since $\eta$ is uniformly distributed, we have $\int_0^1 \int_0^1\theta_j\left( Q_{\epsilon|\eta}(\tau|v) \right) d\tau dv = \mathbb{E}[\theta_j(\epsilon)]$.
\end{Example}
%%%%%%%%%%%%%%%%%%%%%%%%%%%%%%%%%%%%%%%%%%%

%%%%%%%%%%%%%%%%%%%%%%%%%%%%%%%%%%%%%%%%%%%
\section{Calculating the bounds using real data}

In this section, we compute the bounds defined in Theorem 1 using the data in \cite{macours2012cash} and show that our bounds are informative.
Specifically, \cite{macours2012cash} analyze the effect of income on early childhood cognitive development using the \textit{Atenci\'{o}n a Crisis} program, a cash transfer program implemented in rural areas of Nicaragua.
As in Example 1, we focus on income effects on early childhood cognitive development.

In the analysis, we use only children between five and seven years to control for age effects.
The sample size for this analysis is 447, the size of the treatment group is 206, and that of the control group is 241.
Following \cite{macours2012cash}, we use a standardized test score of receptive vocabulary (TVIP) as the outcome of a child's cognitive development.
The average test score is 0.449 and the standard deviation 1.212.
We use the logarithm of total consumption per capita as the endogenous explanatory variable $X$ and the control indicator as the instrument $Z$.
The OLS and IV estimates of the effect of $X$ on $Y$ are 0.592 and 0.841, respectively.
We assume that the effect of income on a child's cognitive development is nonnegative.
Hence, we assume the monotonicity of the structural function and compute the lower and upper bounds under monotonicity.

We estimate the conditional distribution and quantile functions, $F_{Y|X,Z}$, $F_{X|Z}$, $Q_{Y|X,Z}$, and $Q_{X|Z}$, and compute the bounds defined in Theorem 1 by treating these estimates as true functions.
Figure 12 shows the estimates of $F_{X|Z}(x|0)$ and $F_{X|Z}(x|1)$.
Because these functions do not have any intersections, Assumption 3 (i) is satisfied.
Although the distance of the conditional distribution functions seem to be close at the endpoints of the support, the distance of the conditional quantile functions is not close to 0.
Indeed, we have $Q_{X|Z}(0.99|0)-Q_{X|Z}(0.99|1)=0.122$ and $Q_{X|Z}(0.01|0)-Q_{X|Z}(0.01|1)=0.610$.
In addition, the empirical supports of $X|Z=0$ and $X|Z=1$ are $[7.605, 10.010]$ and $[5.900,9.850]$, and hence the boundaries of the empirical supports satisfy Assumption 3 (ii).
Since the estimates of the tail of the probability distributions are unreliable, we only use the estimates of $F_{X|Z}(x|0)$ and $F_{X|Z}(x|1)$ between 0.1 and 0.9, and compute $B^{ML}(x,0.5)$ and $B^{MU}(x,0.5)$ from these estimates.
As shown in Figure 13, the bounds imply that our identification approach can provide informative bounds.
The average difference between $B^{ML}(x,0.5)$ and $B^{MU}(x,0.5)$ is 0.045, which is small compared with the standard deviation of $Y$.
Figure 13 also shows that the structural function $g(x,0.5)$ is close to flat when $x$ is low.
In view of Theorem 3, this fact contributes to narrowing the bounds on $g(x,0.5)$.

Figures 14 and 15 show the lower and upper bounds of $g(8,e)$ and $g(9,e)$ over $e \in [0.25,0.75]$.
As shown in these figures, the lower and upper bounds cross.
There are the following two possible reasons: (i) the assumptions do not hold for higher $e$ and (ii) the estimated functions have sampling errors.
For the first possible reason, the monotonicity assumption may not hold at higher quantiles.
If Assumption 6 is not satisfied for some $e$, then $B^{ML}(x,e)$ may be larger than $B^{MU}(x,e)$.
This result implies that Assumption 6 is testable.
The second possible reason is that we treat the estimates of the conditional distributions and quantiles as true functions.
If the true lower and upper bounds are close, that is, the structural function is nearly point identified, then computed lower and upper bounds may cross.

Using the bounds, $B^{ML}(x,e)$ and $B^{MU}(x,e)$, we compute the bounds of $\Delta g(x,e) \equiv \{g(x+0.25,e)-g(x-0.25,e)\}/0.5$ for $x=8, \, 9$ and $e=0.25, \, 0.5$.
As shown in Figures 14 and 15, the bounds of $g(x,0.75)$ are unreliable.
Hence, we do not compute the bounds of $\Delta g(x,e)$ for $e=0.75$.
We obtain the following lower and upper bounds:
\begin{eqnarray}
0.069 \ \leq & \Delta g(8,0.25) & \leq \ 0.208, \nonumber \\
0.110 \ \leq & \Delta g(8,0.5) & \leq \ 0.252, \nonumber \\
1.181 \ \leq & \Delta g(9,0.25) & \leq \ 1.308, \nonumber \\
2.020 \ \leq & \Delta g(9,0.5) & \leq \ 2.266. \nonumber 
\end{eqnarray}
For low-income ($x=8$) households, the effects of income on a child's cognitive development are small at both the middle and the lower quantiles.
On the contrary, for high-income ($x=9$) households, the effects are large at both quantiles and the impact at the middle quantile is approximately twice as large as that at the lower quantile.
Hence, for high-income households, $\Delta g(x,e)$ is larger than the OLS (or IV) estimate at both the middle and the lower quantiles.
These results imply that the effect of income on a child's cognitive development is quite small for low-ability children from low-income households.

If we consider the model $Y = \beta_0 + \beta_1 X + \beta_2 X^2 + U$, we may capture the nonlinearity with respect to $X$.
However because $Z$ is binary, we cannot estimate this model using the conventional IV estimator.
In addition, classical additive models cannot capture the unobserved heterogeneity in the effect of $X$ on $Y$.
On the contrary, our approach can capture the nonlinearity with respect to $X$ and unobserved heterogeneity.

%%%%%%%%%%%%%%%%%%%%%%%%%%%%%%%%%%%%%%%%%%%
\section{Conclusions}

In this study, we explored the partial identification of nonseparable models with continuous endogenous and binary instrumental variables.
We showed that the structural function is partially identified when it is monotone or concave in the explanatory variable.
\cite{d2015identification} and \cite{torgovitsky2015identification} prove the point identification of the structural function under a key assumption that the conditional distribution functions of the endogenous variable for different values of the instrumental variables have intersections.
We demonstrated that, even if this assumption does not hold, monotonicity and concavity provide identifying power.
Point identification was achieved when the structural function is flat or linear with respect to the explanatory variable over a given interval.
We computed the bounds using real data and showed that our bounds are informative.

\newpage
\renewcommand{\theequation}{A.\arabic{equation}}
\setcounter{equation}{0}
%%%%%%%%%%%%%%%%%%%%%%%%%%%%%%%%%%%%%%%%%%%
\section*{Appendix 1: Proofs}

%%% Prop 1 %%%
\begin{proof}[Proof of Proposition 1]
\underline{Step.1}\,
We show that, for all $e \in (0,1)$ and $x \in \mathcal{X}_0$,
\begin{eqnarray}
P(\epsilon \leq e | X=x, Z=0) = P(\epsilon \leq e | X=\pi(x), Z=1). \label{dist_e}
\end{eqnarray}

First, we examine variable $V \equiv F_{X|Z}(X|Z)$.
This is called ``control variable'' in \cite{imbens2009identification}.
Let $h^{-1}(z,x)$ be the inverse function of $h(z,v)$ with respect to $v$.
We thus have, for all $(z,x) \in \{0,1\} \times \mathcal{X}_z$,
\begin{eqnarray}
 F_{X|Z}(x|z) &=& P\left( h(z,\eta) \leq x | Z=z\right) \nonumber \\
 &=& P\left( \eta \leq h^{-1}(z,x) | Z=z\right) \nonumber \\
 &=& P\left( \eta \leq h^{-1}(z,x) \right) = h^{-1}(z,x), \nonumber
\end{eqnarray}
where the second equality follows from the strict monotonicity of $h(x,v)$ in $v$ and the third equality follows from $Z \indepe (\epsilon,\eta)$.
Therefore, we obtain
$$
V = h^{-1}(Z,X)=\eta.
$$

Next, we show that the conditional distribution of $\epsilon$ conditional on $(X,Z)=(x,z)$ is the same as that of $\epsilon$ conditional on $V=F_{X|Z}(x|z)$.
Because $(x,z) \rightarrow (F_{X|Z}(x|z),z)$ is one-to-one and $F_{X|Z}(x|z)$ is continuous in $x$, the $\sigma$-field generated by $X$ and $Z$ is the same as that generated by $V$ and $Z$.
Hence, we have
$$
P\left( \epsilon \leq e | X=x, Z=z \right) = P\left( \epsilon \leq e | V=F_{X|Z}(x|z), Z=z \right).
$$
It follows from $Z \indepe (\epsilon, \eta)$ and $V=\eta$ that
\begin{eqnarray}
 P\left( \epsilon \leq e | X=x, Z=z \right) &=&  P\left( \epsilon \leq e | V=F_{X|Z}(x|z) \right). \label{V_property} 
\end{eqnarray}
Hence, the conditional distribution of $\epsilon$ conditional on $X$ and $Z$ solely depends on $V=F_{X|Z}(X|Z)$.

By definition, functions $\pi(x)$ and $\pi^{-1}(x)$ satisfy
\begin{equation}\label{pi_property2} 
\begin{aligned}
F_{X|Z}(\pi(x)|1) &=& F_{X|Z}(x|0), \\
F_{X|Z}(\pi^{-1}(x)|0) &=& F_{X|Z}(x|1).
\end{aligned}
\end{equation}
Hence, events $\{X=x,Z=0\}$ and $\{X=\pi(x),Z=1\}$ have the same $V=F_{X|Z}(X|Z)$, and (\ref{dist_e}) follows from (\ref{V_property}). \\
\underline{Step.2}\,
We show that (\ref{dist_e}) implies $g\left( \pi(x),e \right) = \tilde{T}_{x,1} \left( g(x,e) \right)$.
For all $(x,e) \in \mathcal{X}_0 \times (0,1)$, we have
\begin{eqnarray}
\tilde{T}_{x,1} \left( g(x,e) \right) &=& Q_{Y|X=\pi(x),Z=1} \left( F_{Y|X=x,Z=0} \left( g(x,e) \right) \right) \nonumber \\
 &=& Q_{Y|X=\pi(x),Z=1}\left( P(\epsilon \leq e | X=x, Z=0) \right) \nonumber \\
 &=& Q_{Y|X=\pi(x),Z=1}\left( P(\epsilon \leq e | X=\pi(x), Z=1) \right) \nonumber \\
 &=& Q_{Y|X=\pi(x),Z=1}\left( F_{Y|X=\pi(x),Z=1}\left( g(\pi(x),e) \right) \right) \nonumber \\
 &=& g(\pi(x),e), \nonumber
\end{eqnarray}
where the third equality follows from (\ref{dist_e}).
Similarly, we can prove $g\left( \pi^{-1}(x),e \right) = \tilde{T}_{x,-1} \left( g(x,e) \right)$.
\end{proof}
\vspace{0.1in}

%%% Lem 1 %%%
\begin{proof}[Proof of Lemma 1]
Observe that, if $\pi^n(x)$ exists and $\pi^n(x) \in \mathcal{X}_0$, then $\pi^{n+1}(x)$ also exists from (\ref{pi}).
Suppose that there does not exist $n \in \mathbb{N}\cup \{0\}$ such that $\pi^n(x)\in \mathcal{X}_1 \cap \mathcal{X}_0^c = (\underline{x}_1,\underline{x}_0]$.
Then, there exists sequence $\{ x_n \}_{n=0}^{\infty}$ such that $x_n = \pi^n(x)$.
By (\ref{pi_property1}), $\{ x_n \}_{n=0}^{\infty}$ is a decreasing sequence.
Because $x_n > \underline{x}_0$, $\{x_n\}_{n=0}^{\infty}$ converges to $x_{\infty} \in [\underline{x}_0, \overline{x}^0)$.
It follows from (\ref{pi_property2}) that
$$
 F_{X|Z}(x_{n+1}|1) = F_{X|Z}(x_n|0),
$$
meaning we have $F_{X|Z}(x_{\infty}|1) = F_{X|Z}(x_{\infty}|0)$ by the continuity of $F_{X|Z}$. 
However, this equation violates Assumption 3.
Hence, for all $x \in \mathcal{X}$, there exists $n \in \mathbb{N}\cup\{0\}$ such that $\pi^n(x) \in \mathcal{X}_1 \cap \mathcal{X}_0^c$.
Consequently, $\pi^{n'}(x)$ does not exist for $n'>n$.
Similarly, for all $x \in \mathcal{X}$, we have $\pi^{-m}(x) \in \mathcal{X}_0 \cap \mathcal{X}_1^c$ for some $m \in \mathbb{N}\cup\{0\}$.
Then, $\pi^{-m'}(x)$ does not exist for $m'>m$.
Therefore, $\Pi_{x',x}^M$ is finite for all $(x,x') \in \mathcal{X} \times \mathcal{X}$ because the set $\{(n,m)\in \mathbb{Z} \times \mathbb{Z}: \pi^n(x')\ \text{and} \ \pi^m(x) \ \text{exist.} \}$ is finite.

We proceed to show the nonemptiness of $\Pi_{x',x}^M$.
For all $x,x' \in \mathcal{X}$, $(n,m)\in \mathbb{Z}\times\mathbb{Z}$ exists such that $\pi^{n}(x') \in \mathcal{X}_1 \cap \mathcal{X}_0^c = (\underline{x}_1,\underline{x}_0]$ and $\pi^{m}(x) \in \mathcal{X}_0 \cap \mathcal{X}_1^c = [\overline{x}_1,\overline{x}_0)$.
It follows from Assumption 3 (ii) that $\pi^{n}(x') < \pi^{m}(x)$.
\end{proof}
\vspace{0.1in}

%%% Thm 1 %%%
\begin{proof}[Proof of Theorem 1]
As discussed in Section 3, it suffices to show that $T_{x',x}^{ML}(y)$ and $T_{x',x}^{MU}(y)$ are strictly increasing in $y$ and surjective.
If $\pi^n(x)$ exists, $\tilde{T}_{x,n}(y)$ is strictly increasing in $y$.
Hence, $T_{x',x}^{ML}(y)$ and $T_{x',x}^{MU}(y)$ are strictly increasing in $y$ because $\Pi_{x',x}^M$ is finite by Lemma 1.
If $\pi^n(x)$ exists, we obtain $\tilde{T}_{x,n}([\underline{y},\overline{y}]) = [\underline{y},\overline{y}]$.
Hence, because $\Pi_{x',x}^M$ is finite, we have $T_{x',x}^{ML}(y)$ and $T_{x',x}^{MU}(y)$ are  surjective.
\end{proof}
\vspace{0.1in}

%%% Lem 2 %%%
\begin{proof}[Proof of Lemma 2]
From the proof of Lemma 1, $\Pi^C_{x',x}$ is finite.
Hence, we prove the nonemptiness of $\Pi^C_{x',x}$.
From the proof of Lemma 1, for all $x,x' \in \mathcal{X}$, there exist $n, m \in \mathbb{Z}$ such that $\pi^{m}(x), \pi^{n}(x') \in \mathcal{X}_1 \cap \mathcal{X}_0^c$.
Without loss of generality, we assume $\pi^{n}(x') \leq \pi^{m}(x)$.
Then, $\pi^{m-1}(x)$ and $\pi^{n-1}(x')$ exist because $\pi^{m}(x), \pi^{n}(x') \in \mathcal{X}_1$.
Because $\pi^m(x) \in \mathcal{X}_1 \cap \mathcal{X}_0^c$ and $\pi^{n-1}(x') \in \mathcal{X}_0$, we have $\pi^{n}(x') \leq\pi^{m}(x) \leq \pi^{n-1}(x') \leq \pi^{m-1}(x)$ from (\ref{pi_property3}), and hence $(n,m) \in \Pi_{x',x}^C$.
Therefore, $\Pi_{x',x}^C$ is nonempty.
\end{proof}
\vspace{0.1in}

%%% Thm 2 %%%
\begin{proof}[Proof of Theorem 2]
Similar to the proof of Theorem 1, we can obtain 
$$
\left( G_x^{CL} \right)^{-1}(e) \leq g(x,e) \leq \left( G_x^{CU} \right)^{-1}(e).
$$
Because $g(x,e)$ is concave in $x$, if $x = ty' + (1-t)y$ and $t \in (0,1)$, then we have $g(x,e) \geq t g(y',e) + (1-t) g(y,e) \geq t \left( G_{y'}^{CL} \right)^{-1}(e) + (1-t) \left( G_{y}^{CL} \right)^{-1}(e)$.
Hence, we have
$$
g(x,e) \geq \sup_{y,y':y < x < y'} \left\{ \left( \frac{x-y}{y'-y} \right) \left( G_{y'}^{CL} \right)^{-1}(e) + \left( \frac{y'-x}{y'-y} \right) \left( G_y^{CL} \right)^{-1}(e) \right\}.
$$
Because $g(x,e)$ is concave in $x$, if $x = ty' + (1-t)y$ and $t < 0$, then we have $g(x,e) \leq t g(y',e) + (1-t) g(y,e)$.
Because $B^{CL}(x,e) \leq g(x,e) \leq \left( G_x^{CU} \right)^{-1}(e)$, $t<0$, and $1-t>0$, we have $g(x,e) \leq t  B^{CL}(y',e) + (1-t) \left( G_{y}^{CU} \right)^{-1}(e)$.
Similarly, if $x = ty + (1-t)y'$ and $t > 1$, then we have $g(x,e) \leq t g(y,e) + (1-t) g(y',e) \leq t\left( G_{y}^{CU} \right)^{-1}(e) + (1-t) B^{CL}(y',e)$.
Hence, we have
\begin{eqnarray}
g(x,e) \leq & \min \left[  \inf_{y,y':x < y < y'} \left\{ \left( \frac{x-y}{y'-y} \right)B^{CL}(y',e) + \left( \frac{y'-x}{y'-y} \right) \left( G_y^{CU} \right)^{-1}(e) \right\}, \right. \nonumber \\
 & \left. \inf_{y,y':y' < y < x} \left\{ \left( \frac{y-x}{y-y'} \right)B^{CL}(y',e) + \left( \frac{x-y'}{y-y'} \right) \left( G_y^{CU} \right)^{-1}(e) \right\} \right]. \nonumber
\end{eqnarray}
\end{proof}
\vspace{0.1in}

%%% Thm 3 %%%
\begin{proof}[Proof of Theorem 3]
\underline{Step.1}\,
First, we show that, for all $x \in \mathcal{X}$, there exists $n^* \in \mathbb{Z}$ such that $\pi^{n^*}(x)$ and $\pi^{n^*+1}(x)$ are well defined and $\pi^{n^*}(x), \pi^{n^*+1}(x) \in [\pi(\tilde{x}),\pi^{-1}(\tilde{x})]$.
If $\pi^n(x)$ and $\pi^n(y)$ are well defined, because $\pi^n(\cdot)$ is strictly increasing, we can obtain
\begin{equation}
x \leq y \ \ \Rightarrow \ \ \pi^n(x) \leq \pi^n(y). \label{pi_property3}
\end{equation}
We consider the following four cases: (i) $\pi(\tilde{x})\leq x \leq \tilde{x}$, (ii) $\tilde{x} \leq x \leq \pi^{-1}(\tilde{x})$, (iii) $x < \pi(\tilde{x})$, and (iv) $x > \pi^{-1}(\tilde{x})$.
In case (i), it follows from (\ref{pi_property3}) that $\pi(\tilde{x})\leq x \leq \tilde{x} \leq \pi^{-1}(x) \leq \pi^{-1}(\tilde{x})$.
In case (ii), it follows from (\ref{pi_property3}) that $\pi(\tilde{x})\leq \pi(x) \leq \tilde{x} \leq x \leq \pi^{-1}(\tilde{x})$.
In case (iii), it follows from the proof of Lemma 1 that $n \in \mathbb{N}$ exists such that $\pi^{-n}(x) \in \mathcal{X}_0 \cap \mathcal{X}_1^c$.
This implies that $\pi^{-1}(x), ... , \pi^{-n}(x)$ exist.
By the definition of $\pi$, we have $\pi(\tilde{x}) \in \mathcal{X}_1$, and hence $x < \pi(\tilde{x}) < \pi^{-n}(x)$. 
Therefore, there exists $n^* \in \mathbb{Z}$ such that $\pi^{n^*+2}(x) \leq \pi(\tilde{x}) \leq \pi^{n^*+1}(x)$ and we can obtain $\pi(\tilde{x}) \leq \pi^{n^*+1}(x) \leq \tilde{x} \leq \pi^{n^*}(x) \leq \pi^{-1}(\tilde{x})$ from (\ref{pi_property3}).
Similarly, in case (iv), there exists $n^* \in \mathbb{Z}$ such that $\pi^{n^*}(x) , \pi^{n^*+1}(x) \in [\pi(\tilde{x}),\pi^{-1}(\tilde{x})]$.\\
\underline{Step.2}\,
Next, we show that $g$ is point identified.
From step 1, for all $x, x' \in \mathcal{X}$, there exists $n,m \in \mathbb{Z}$ such that $\pi^n(x'),\pi^{n+1}(x'),\pi^m(x),\pi^{m+1}(x) \in [\pi(\tilde{x}),\pi^{-1}(\tilde{x})]$.
Then, from (\ref{pi_property3}), we have either $\pi^{n+1}(x') \leq \pi^{m+1}(x) \leq \pi^{n}(x') \leq \pi^{m}(x)$ or $\pi^{m+1}(x) \leq \pi^{n+1}(x') \leq \pi^{m}(x) \leq \pi^{n}(x')$.
If $\pi^{n+1}(x') \leq \pi^{m+1}(x) \leq \pi^{n}(x') \leq \pi^{m}(x)$, then we have $(n+1,m+1),(n,m) \in \Pi_{x',x}^M$.
If $\pi^{m+1}(x) \leq \pi^{n+1}(x') \leq \pi^{m}(x) \leq \pi^{n}(x')$, then we have $(n+1,m) \in \Pi_{x',x}^M$.
Hence, there exists a pair $(n^*,m^*) \in \Pi^M_{x',x}$ such that $\pi^{n^*}(x'), \pi^{m^*}(x) \in [\pi(\tilde{x}),\pi^{-1}(\tilde{x})]$.
As $g(x,e)$ is constant on $[\pi(\tilde{x}),\pi^{-1}(\tilde{x})]$, we obtain
$$
\tilde{T}_{x',n^*} \left( g(x',e) \right) = \tilde{T}_{x,m^*} \left( g(x,e) \right).
$$
Therefore, $g(x',e)= T_{x',x}^{MU} \left( g(x,e) \right)$.
Hence, $\left( G^{ML}_x \right)^{-1}(e)$ coincides with $g(x,e)$ because (\ref{TU_M}) becomes an equality.
This implies that $B^{ML}(x,e)$ coincides with $g(x,e)$.
Similarly, $B^{MU}(x,e)$ coincides with $g(x,e)$.
\end{proof}
\vspace{0.1in}

%%% Thm 4 %%%
\begin{proof}[Proof of Theorem 4]
Similar to Theorem 3, for all $x,x' \in \mathcal{X}$, there exist $(n,m) \in \Pi^C_{x',x}$ such that $\pi^n(x')$, $\pi^{n-1}(x')$, and $\pi^m(x)$ are in $[\pi(\tilde{x}),\pi^{-1}(\tilde{x})]$.
Because $g(x,e)$ is linear in $x$, we have
$$
t_{x',x}^{n,m} \cdot \tilde{T}_{x',n} \left( g(x',e) \right) + (1-t_{x',x}^{n,m}) \cdot \tilde{T}_{x',n-1} \left( g(x',e) \right)= \tilde{T}_{x,m} \left( g(x,e) \right).
$$
Similarly, for all $x,x' \in \mathcal{X}$, there exist $(n,m) \in \Pi^C_{x,x'}$ such that
$$
\tilde{T}_{x',m} \left( g(x',e) \right) = t_{x,x'}^{n,m} \cdot \tilde{T}_{x,n} \left( g(x,e) \right) + (1-t_{x,x'}^{n,m}) \cdot \tilde{T}_{x,n-1} \left( g(x,e) \right).
$$
Hence, as described above, $B^{CL}(x,e)$ and $B^{CU}(x,e)$ coincide with $g(x,e)$ because inequalities (\ref{TU_C}) and (\ref{TL_C}) become equalities.
\end{proof}
\vspace{0.1in}

\afterpage{\clearpage}
\newpage
%%%%%%%%%%%%%%%%%%%%%%%%%%%%%%%%%%%%%%%%%%%
\section*{Appendix 2: Figures}

\begin{figure}[h]
\centering
\includegraphics[width=15cm]{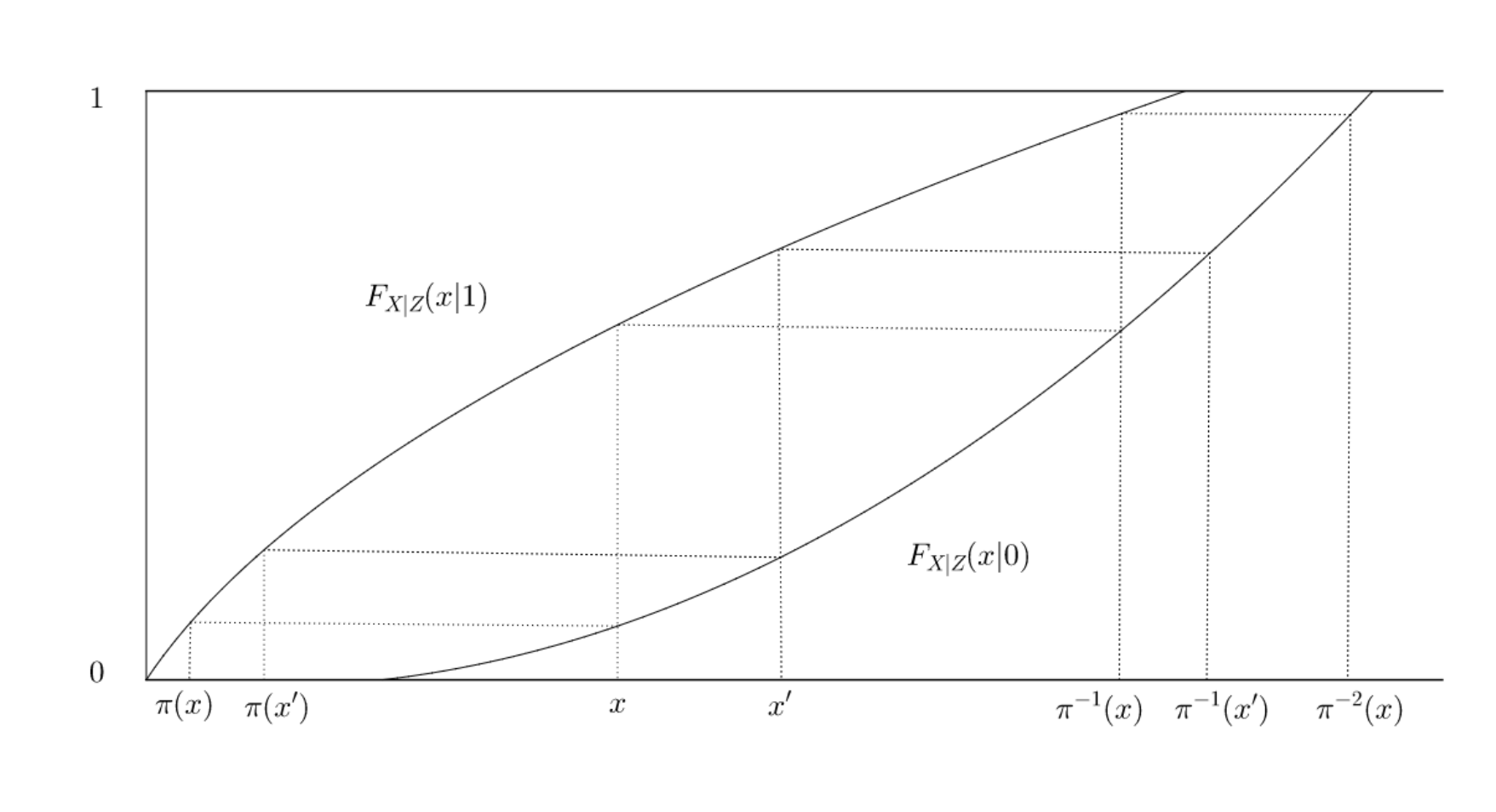}
\caption{The case where Assumption 3 holds.}
\end{figure}

\begin{figure}
  \centering
  \includegraphics[width=15cm]{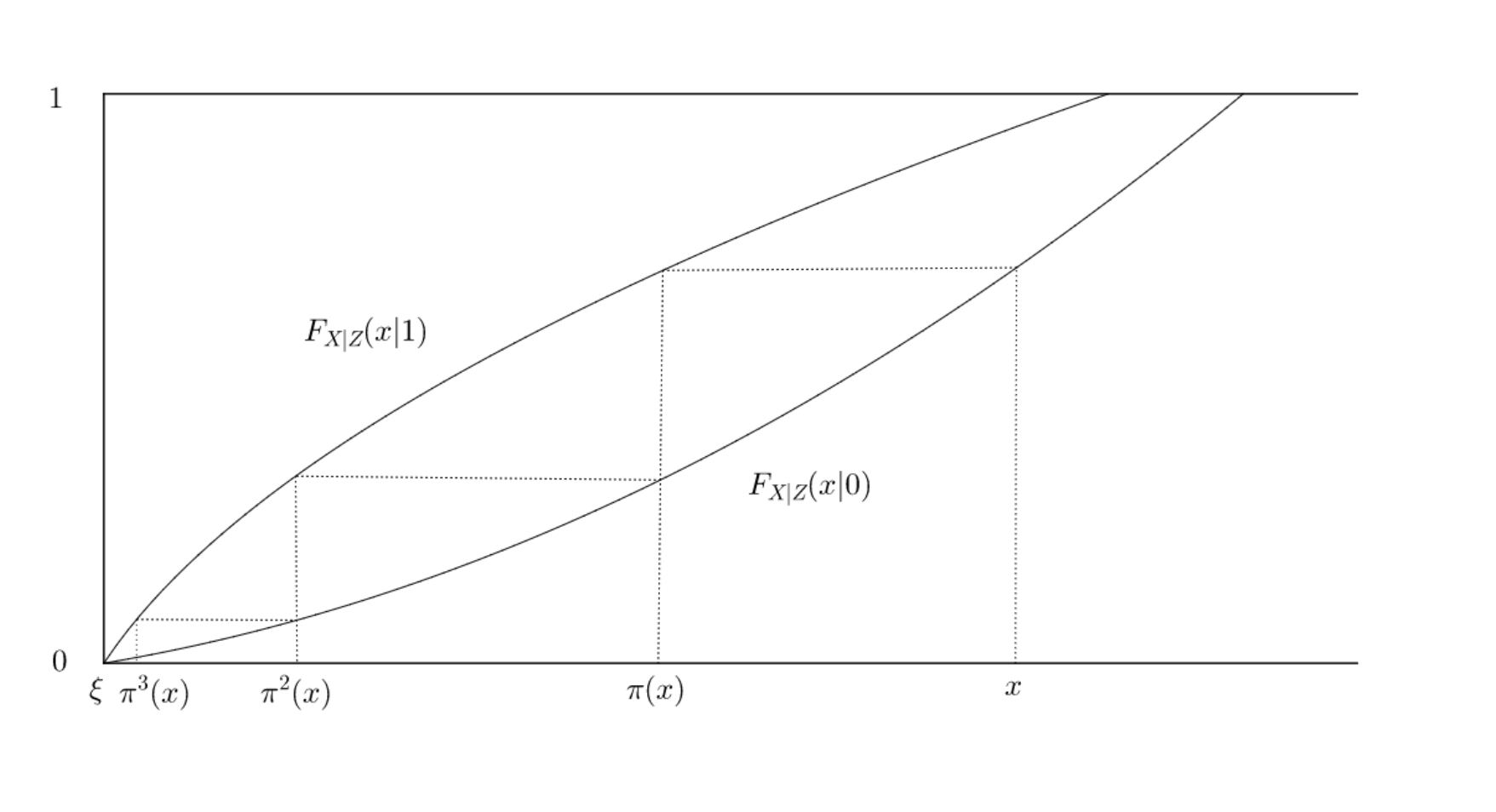}
  \caption{The case where Assumption 3 does not hold.}
\end{figure}

\begin{figure}
  \centering
  \includegraphics[width=16cm]{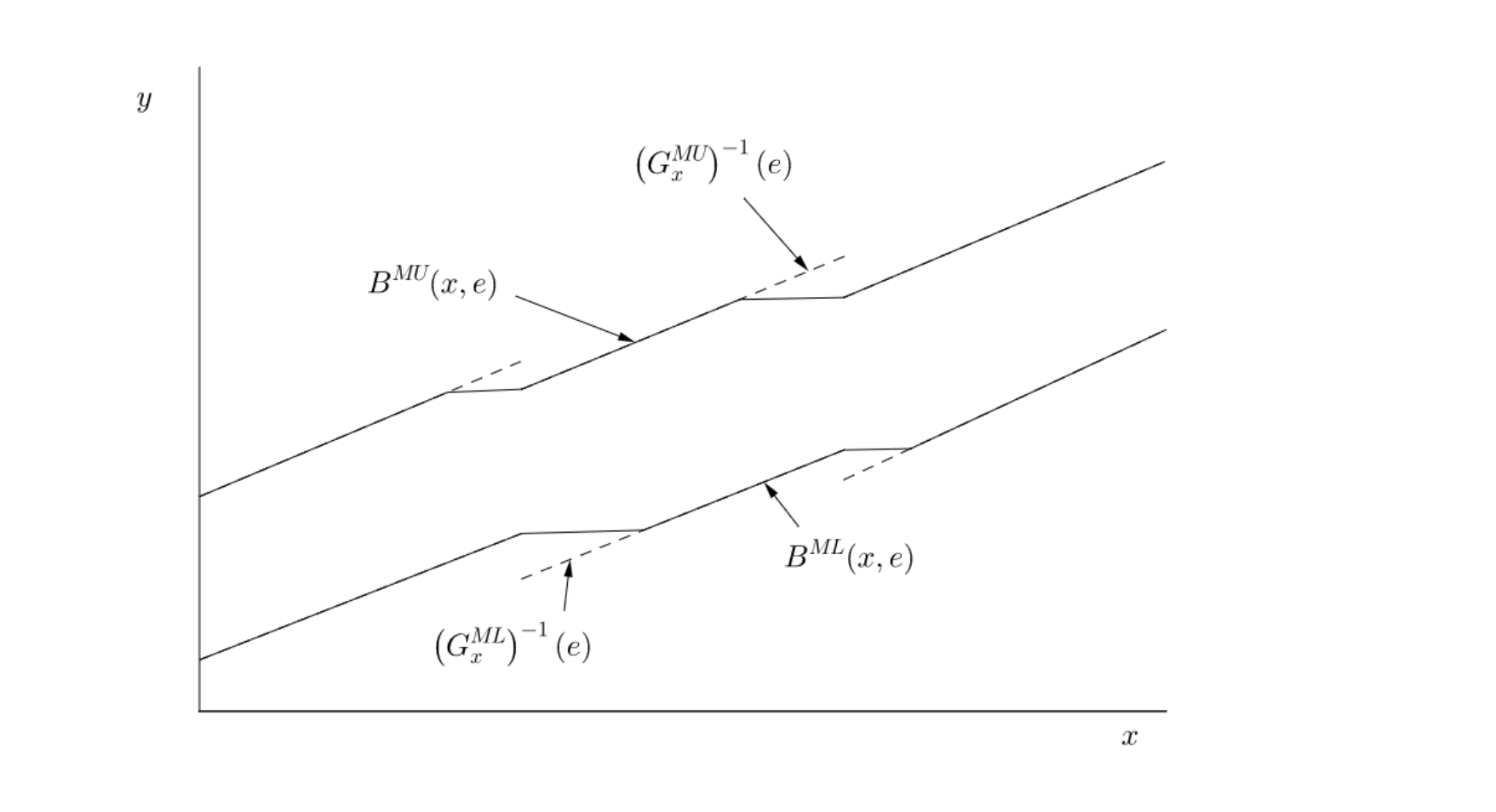}
  \caption{The dashed lines denote $\left( G_{x}^{ML} \right)^{-1}(e)$ and $\left( G_{x}^{MU} \right)^{-1}(e)$. The solid lines denote $B^{ML}(x,e)$ and $B^{MU}(x,e)$.}
\end{figure}

\begin{figure}
  \centering
  \includegraphics[width=16cm]{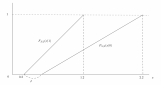}
  \caption{$F_{X|Z}(x|0)$ and $F_{X|Z}(x|1)$ for Simulation 1.}
\end{figure}

\begin{figure}
  \centering
  \includegraphics[width=16cm]{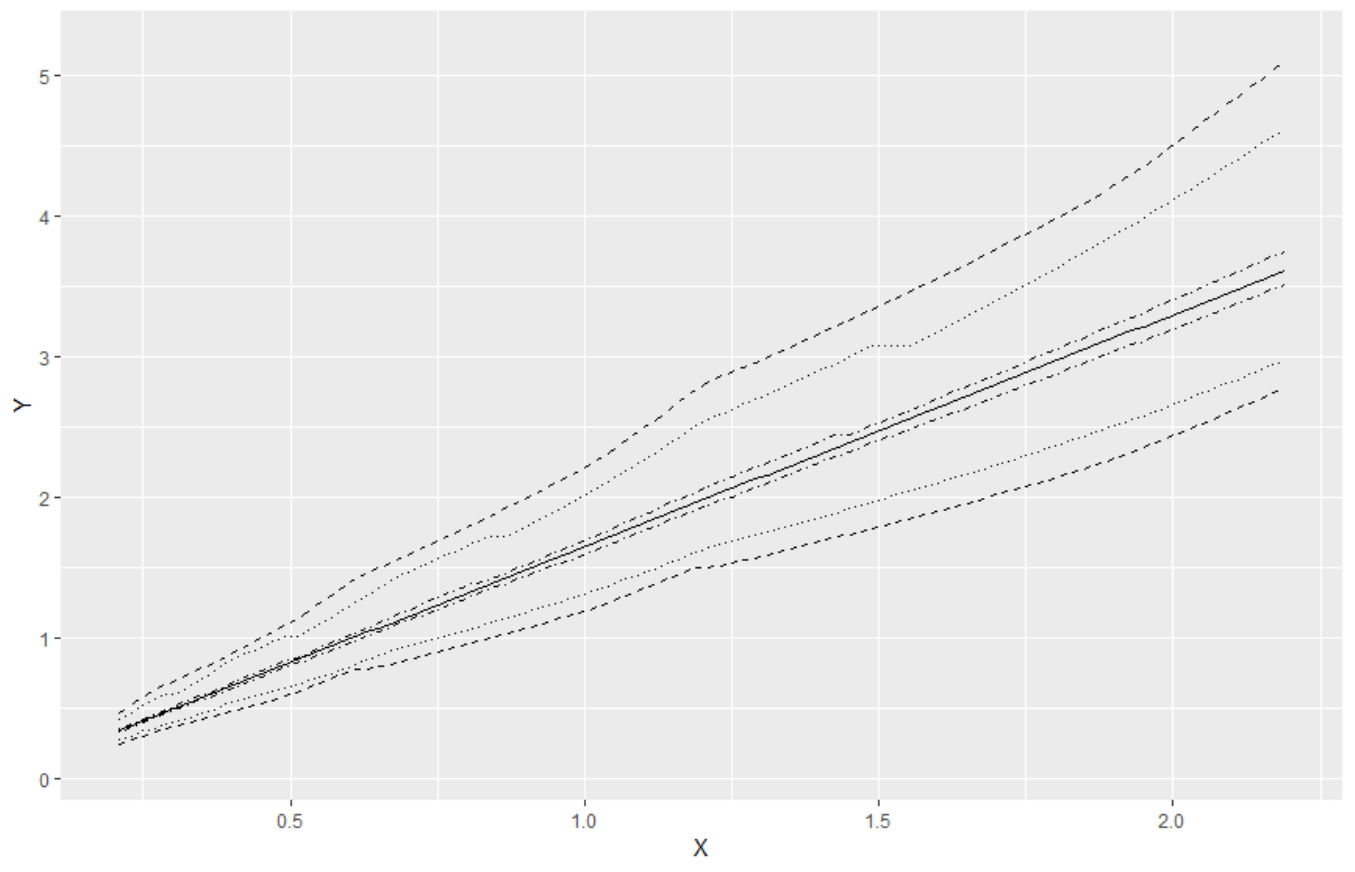}
  \caption{$h(x) = h_1(x)$. The solid line denotes $g(x,0.5)$. The dashed, dotted, and dash-dotted lines denote $B^{ML}(x,0.5)$ and $B^{MU}(x,0.5)$ when $\rho = 0.3$, $0.1$, and $0.01$, respectively.}
\end{figure}

\begin{figure}
  \centering
  \includegraphics[width=16cm]{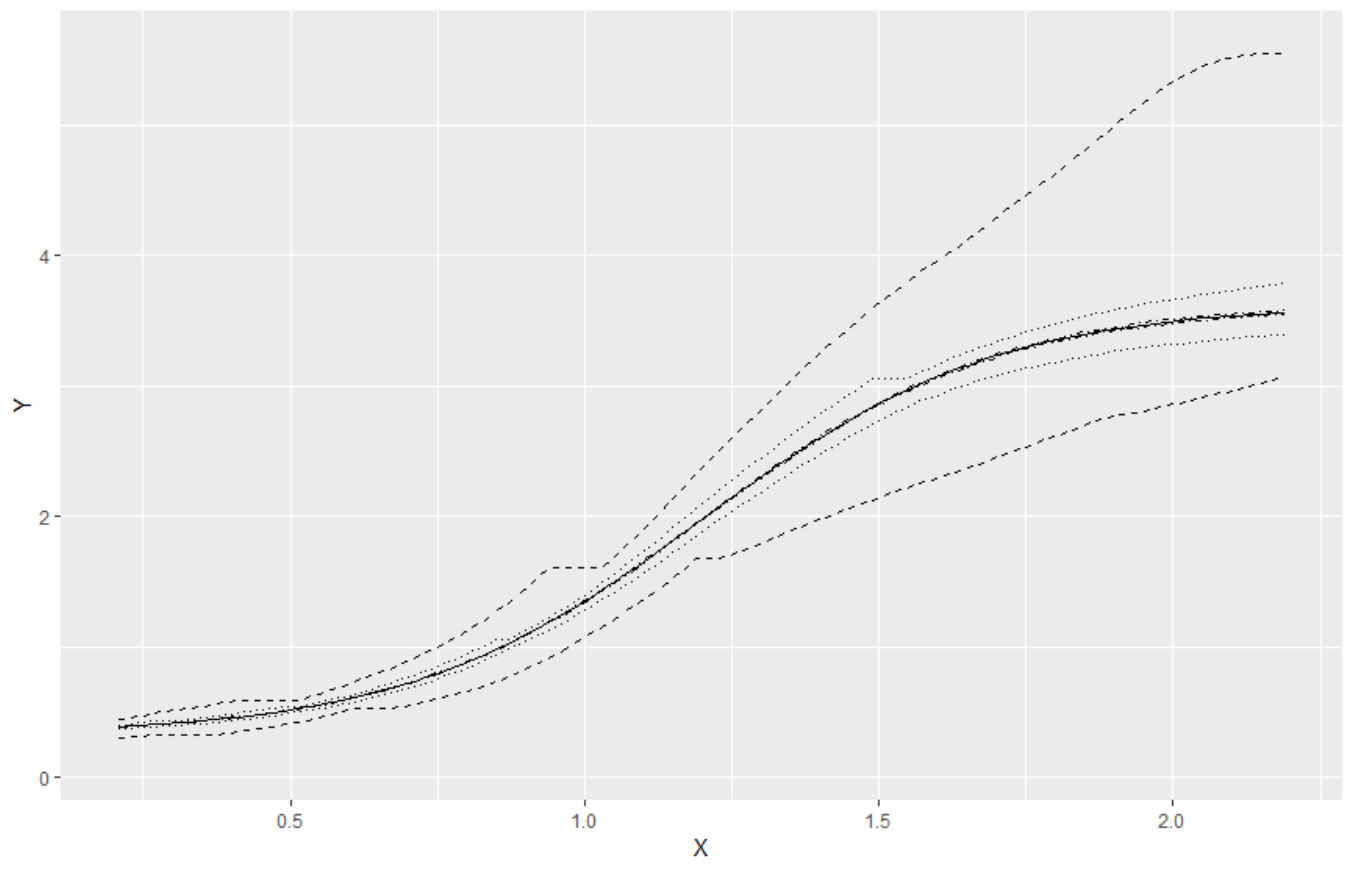}
  \caption{$h(x) = h_2(x)$. The solid line denotes $g(x,0.5)$. The dashed, dotted, and dash-dotted lines denote $B^{ML}(x,0.5)$ and $B^{MU}(x,0.5)$ when $\rho = 0.3$, $0.1$, and $0.01$, respectively.}
\end{figure}

\begin{figure}
  \centering
  \includegraphics[width=16cm]{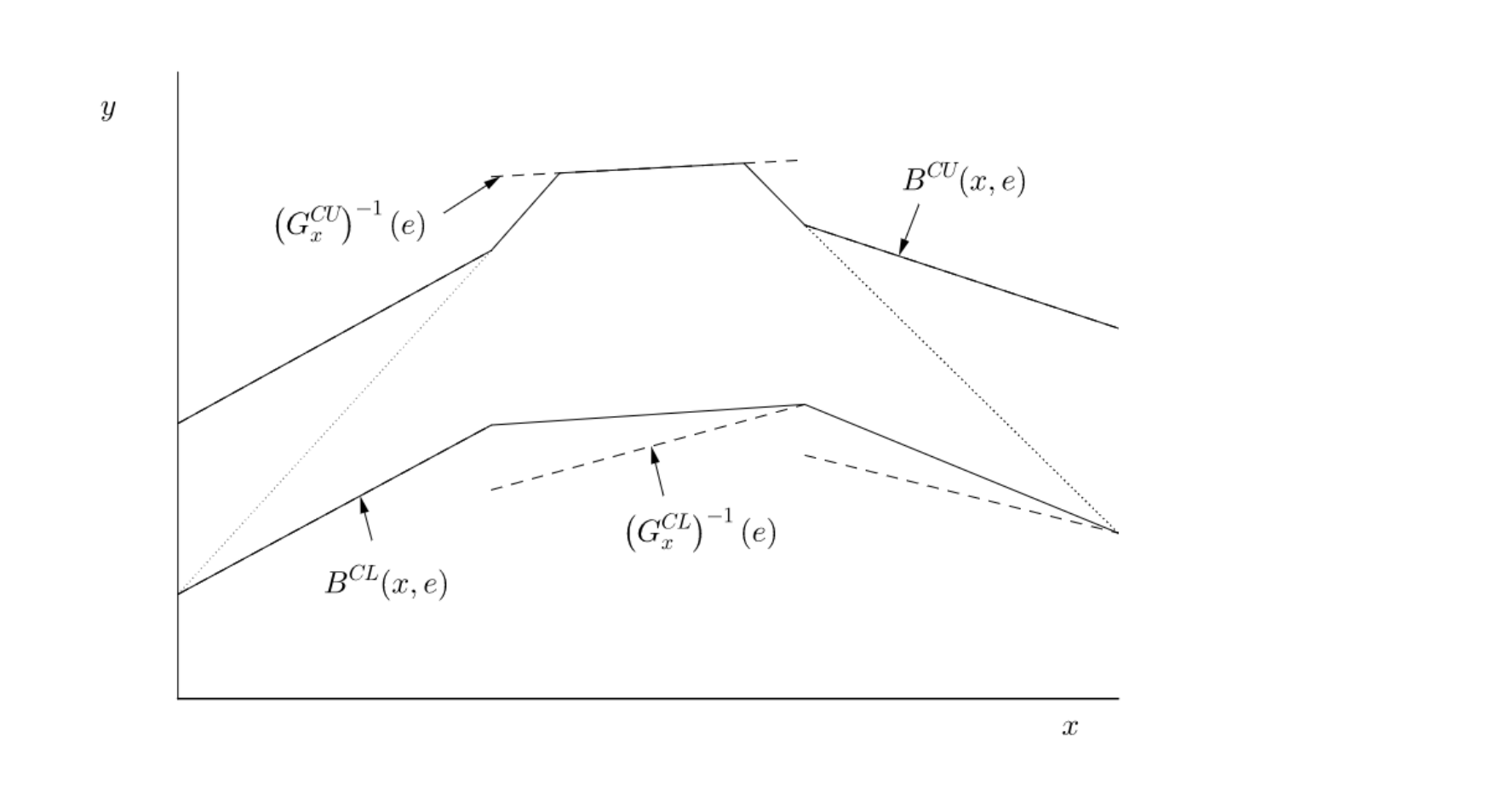}
  \caption{The dashed lines denote $\left( G_{x}^{CL} \right)^{-1}(e)$ and $\left( G_{x}^{CU} \right)^{-1}(e)$. The solid lines denote $B^{CL}(x,e)$ and $B^{CU}(x,e)$.}
\end{figure}

\begin{figure}
  \centering
  \includegraphics[width=15cm]{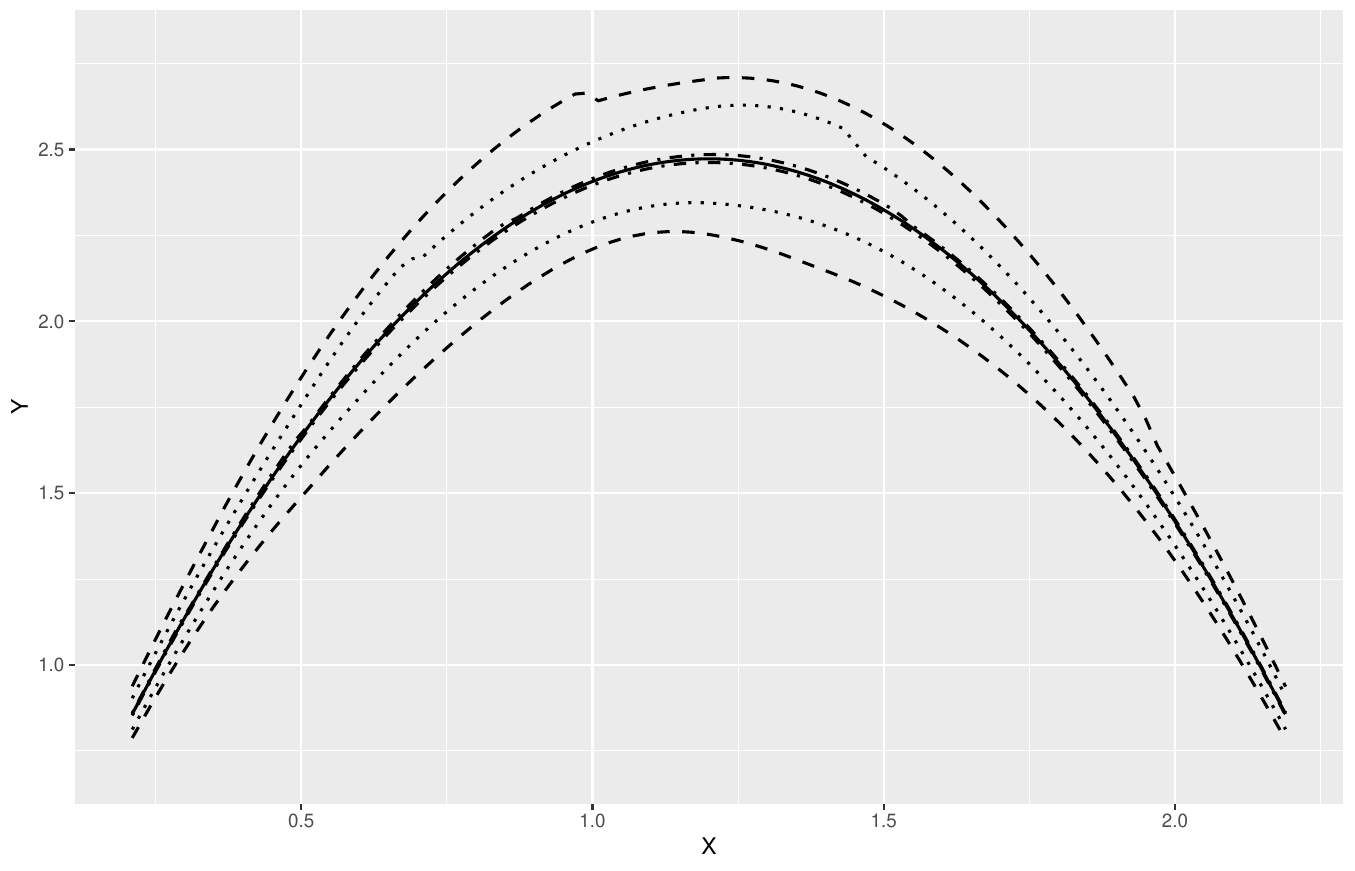}
  \caption{The solid line denotes $g(x,0.5)$. The dashed, dotted, and dash-dotted lines denote $B^{CL}(x,0.5)$ and $B^{CU}(x,0.5)$ when $\rho = 0.8$, $0.5$, and $0.1$, respectively.}
\end{figure}

\begin{figure}
  \centering
  \includegraphics[width=16cm]{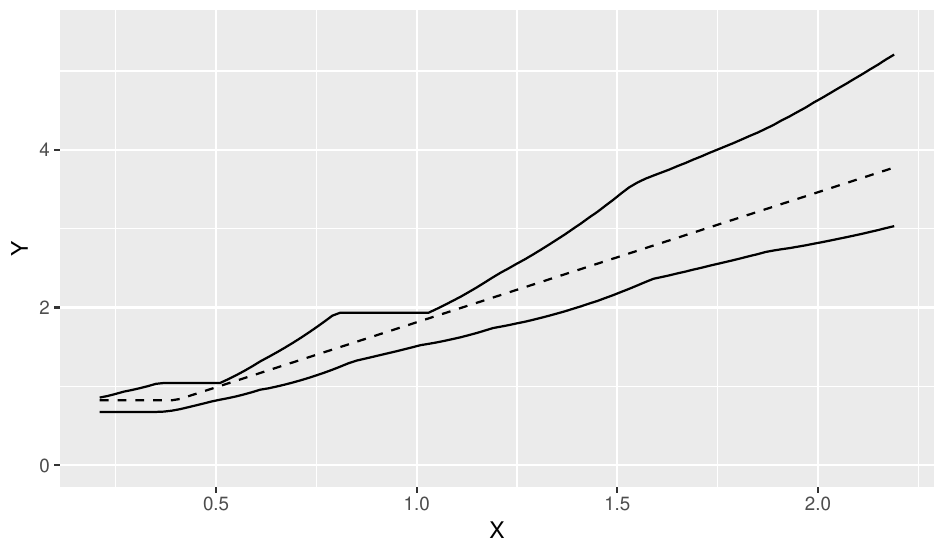}
  \caption{$\delta = 0.4$. The dashed line denotes $g(x,0.5)$. The solid lines denote $B^{ML}(x,0.5)$ and $B^{MU}(x,0.5)$.}
\end{figure}

\begin{figure}
  \centering
  \includegraphics[width=16cm]{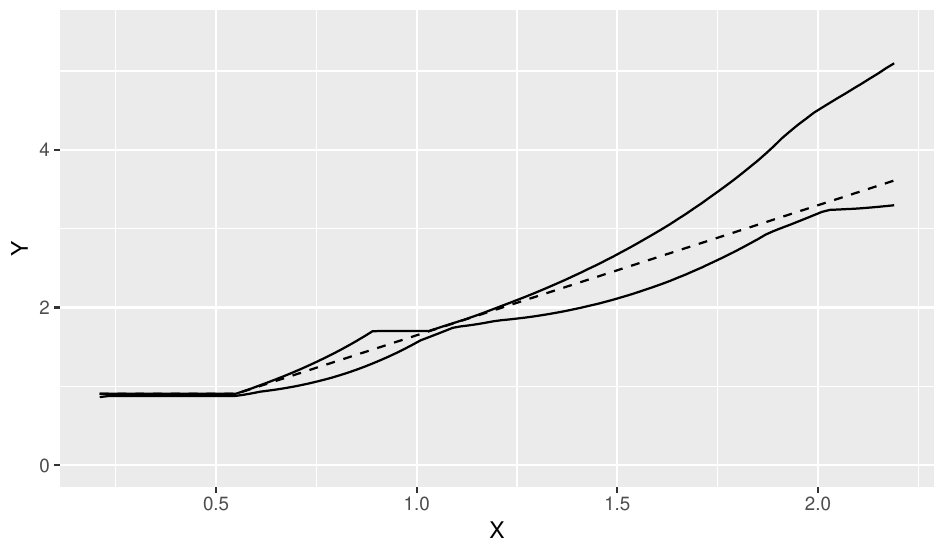}
  \caption{$\delta = 0.55$. The dashed line denotes $g(x,0.5)$. The solid lines denote $B^{ML}(x,0.5)$ and $B^{MU}(x,0.5)$.}
\end{figure}

\begin{figure}
  \centering
  \includegraphics[width=16cm]{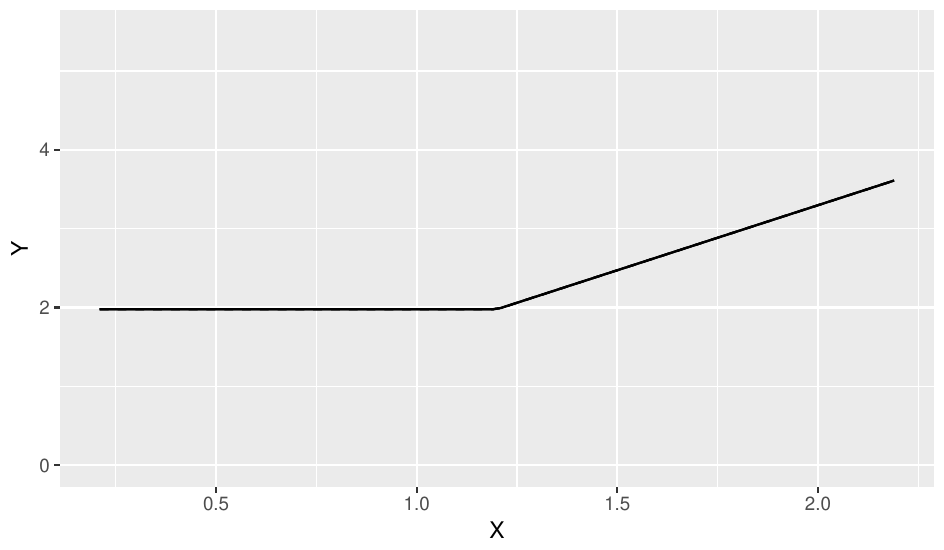}
  \caption{$\delta = 1.2$. The dashed line denotes $g(x,0.5)$. The solid lines denote $B^{ML}(x,0.5)$ and $B^{MU}(x,0.5)$.}
\end{figure}

\begin{figure}
  \centering
  \includegraphics[width=16cm]{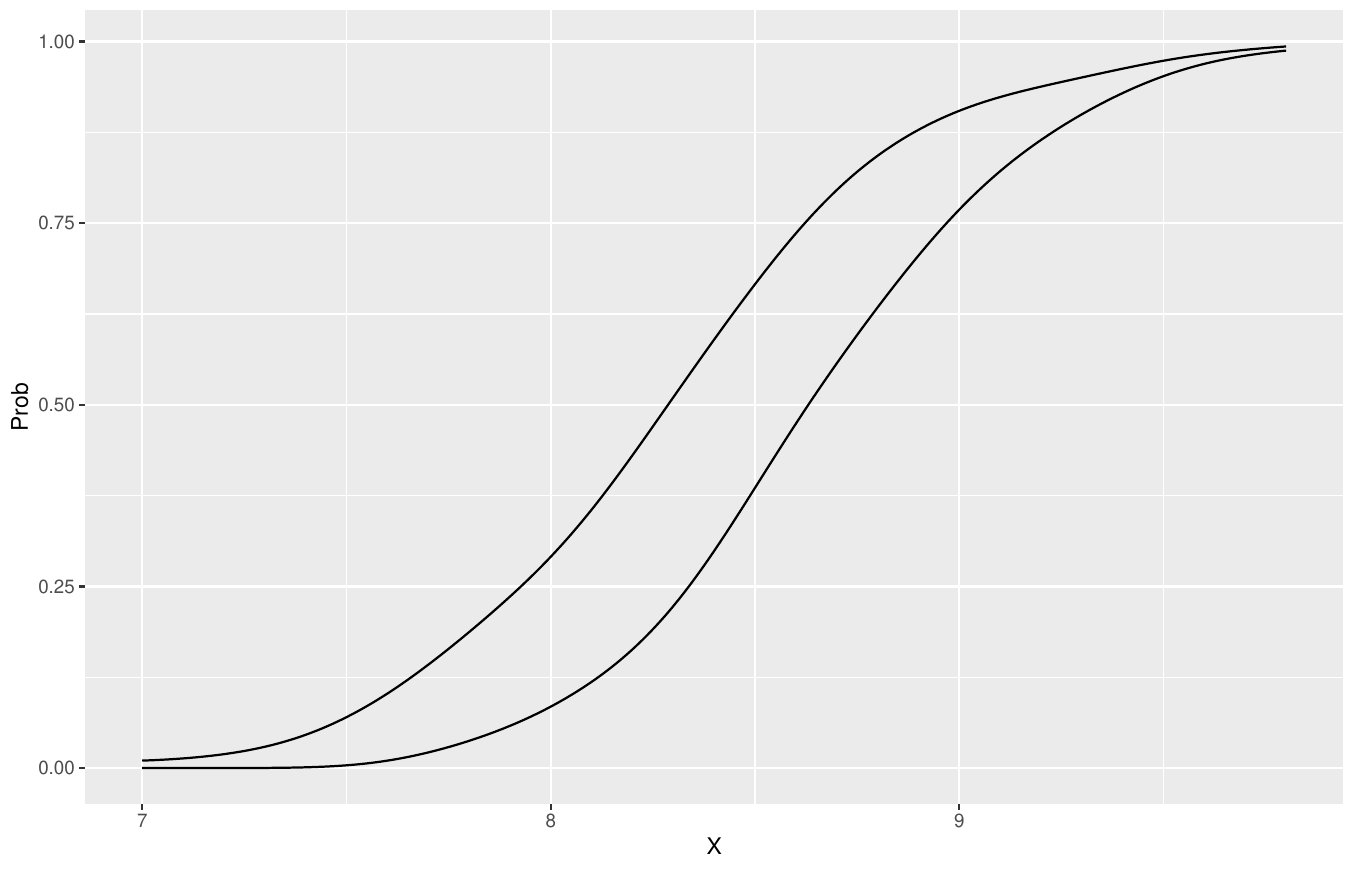}
  \caption{The right-hand line denotes $F_{X|Z}(x|0)$ and the left-hand one denotes $F_{X|Z}(x|1)$.}
\end{figure}

\begin{figure}
  \centering
  \includegraphics[width=16cm]{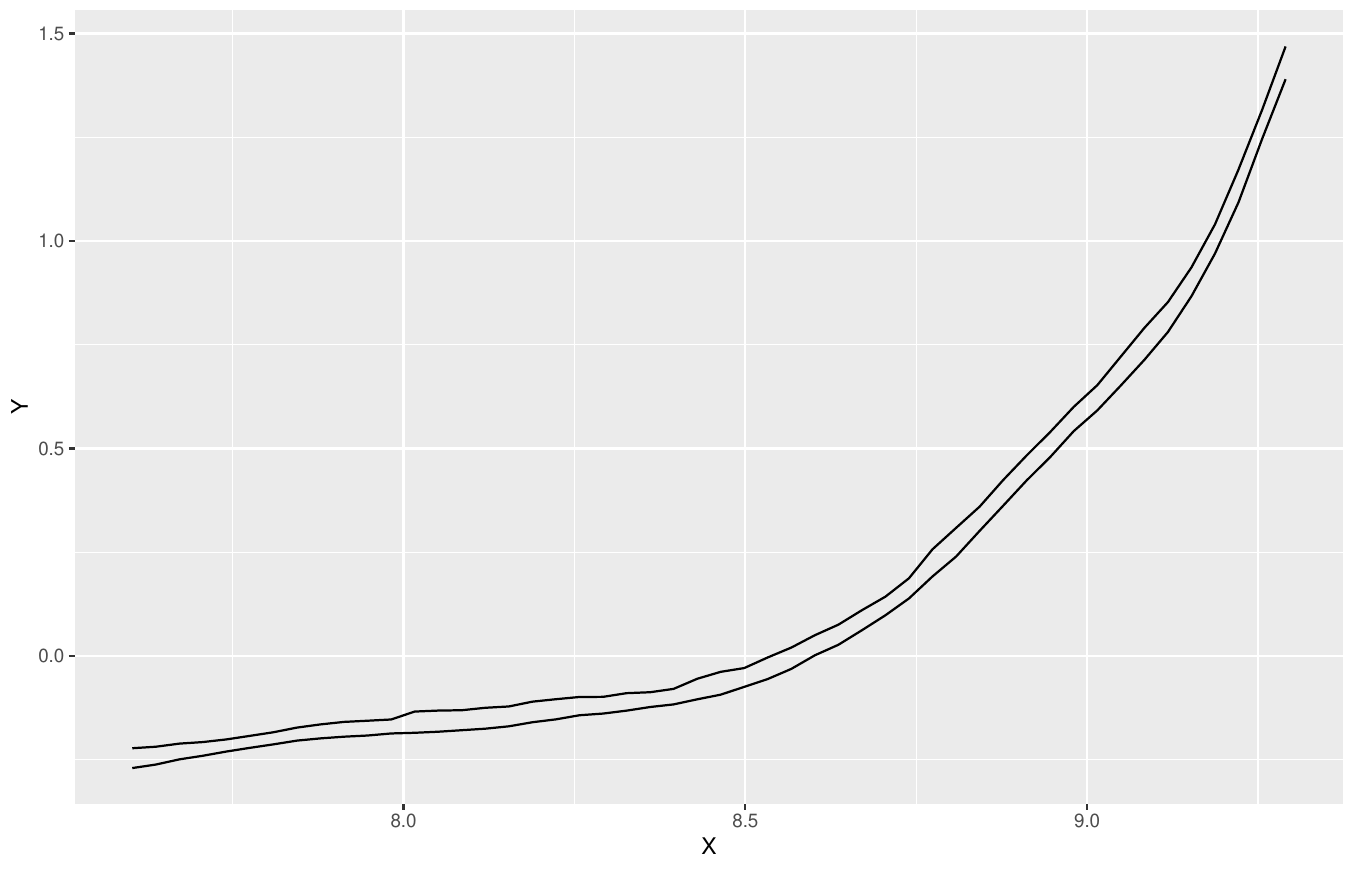}
  \caption{The lower line denotes $B^{ML}(x,0.5)$ and the upper one denotes $B^{MU}(x,0.5)$.}
\end{figure}

\begin{figure}
  \centering
  \includegraphics[width=16cm]{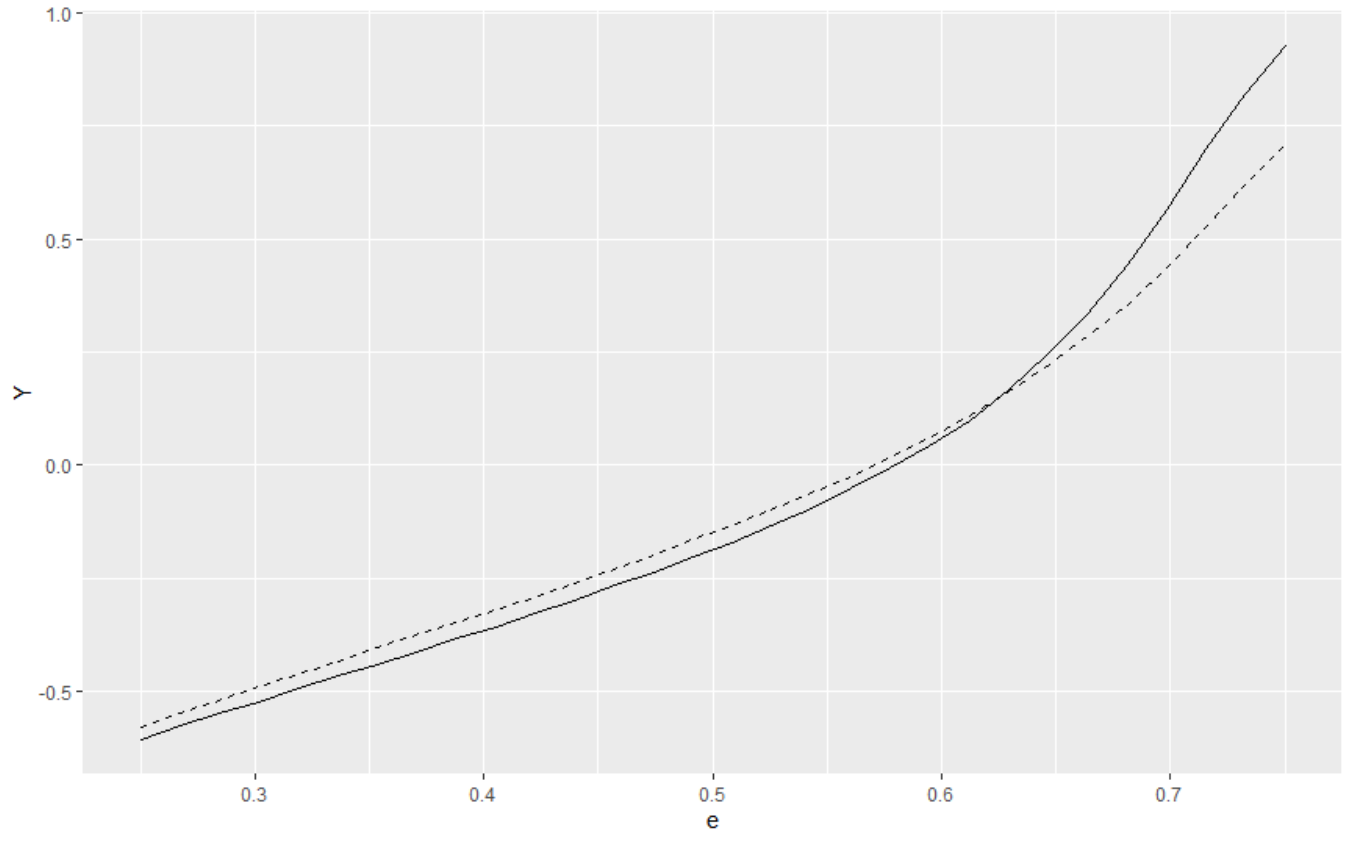}
  \caption{The solid line denotes $B^{ML}(8,e)$ and the dashed line denotes $B^{MU}(8,e)$.}
\end{figure}

\begin{figure}
  \centering
  \includegraphics[width=16cm]{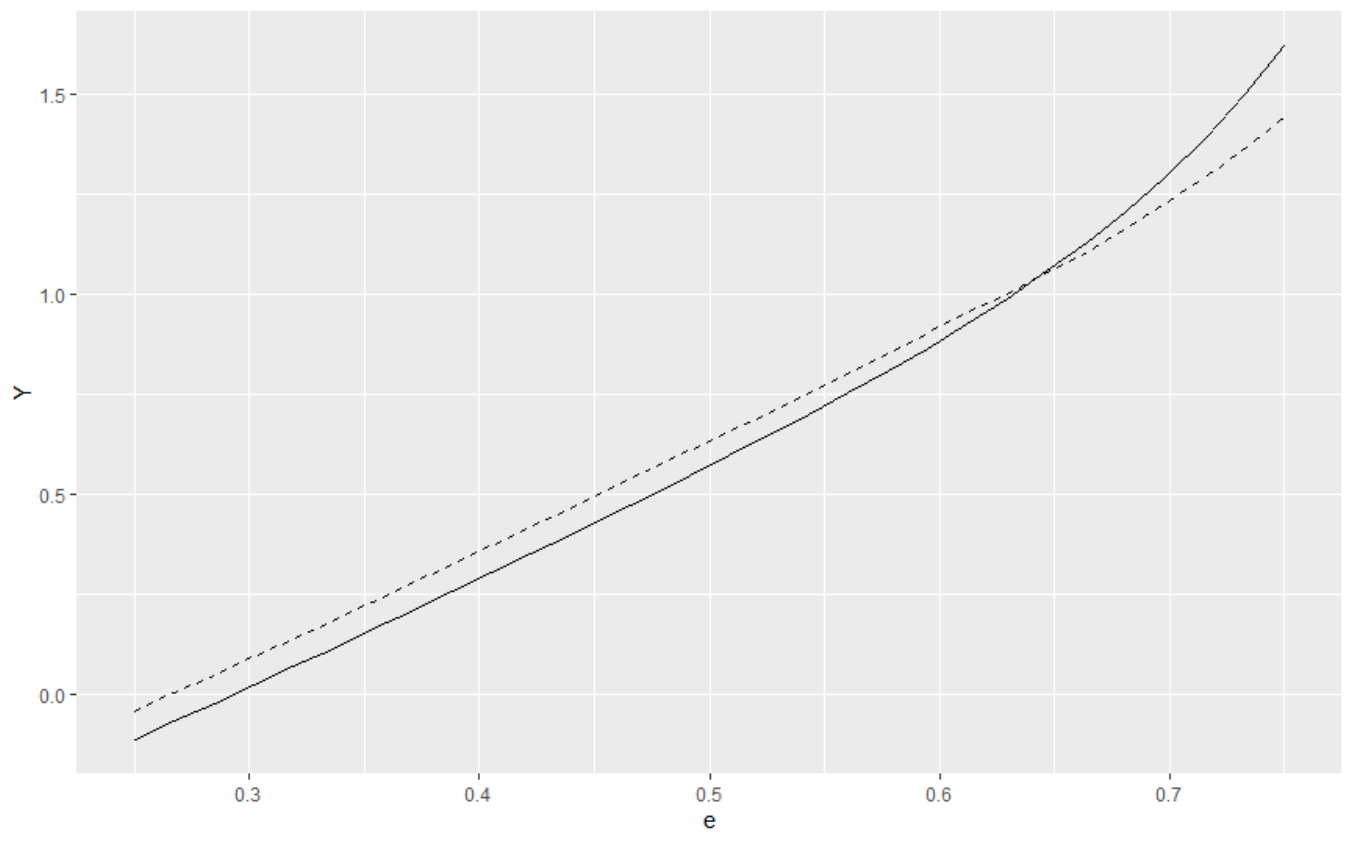}
  \caption{The solid line denotes $B^{ML}(9,e)$ and the dashed line denotes $B^{MU}(9,e)$.}
\end{figure}

\afterpage{\clearpage}
\newpage
%%%%%%%%%%%%%%%%%%%%%%%%%%%%%%%%%%%%%%%%%%%
\section*{Appendix 3: Bounds without Assumption 4 (ii)}

Here, we obtain the lower and upper bounds of $g(x,e)$ under Assumptions 1, 2, 3, 4 (i), 5, and 6.
As such, we can show that $\mathcal{Y}_{x,z}$ is an open interval and does not depend on $z$.
By model (\ref{model}), the support of $Y|X=x,Z=z$ is equivalent to that of $g(x,\epsilon)|X=x,Z=z$.
Hence, under Assumption 5 (ii), we have
$$
\mathcal{Y}_{x,z} = \{g(x,e): e \in (0,1)\},
$$
which implies that $\mathcal{Y}_{x,z}$ does not depend on $z$.
By Assumption 2 (i), $\mathcal{Y}_{x,z}$ must be an open interval.
Hence, we have
$$
\mathcal{Y}_{x,z} = \mathcal{Y}_x \equiv (\underline{y}_x, \overline{y}_x),
$$
where $-\infty \leq \underline{y}_x < \overline{y}_x \leq +\infty$.

First, Proposition 1 holds without Assumption 4 (ii).
Hence, for $n \in \mathbb{Z}$, if $\pi^n(x)$ exists, we can construct $\tilde{T}_x^{(n)} :\mathcal{Y}_x \rightarrow \mathcal{Y}_{\pi^n(x)}$ that satisfies
$$
g\left( \pi^n(x),e \right) = \tilde{T}_{x,n}\left( g(x,e) \right).
$$
If $(n,m) \in \Pi_{x',x}^M$, then Assumption 6 implies that
$$
 \tilde{T}_{x',n} \left( g(x',e) \right) \leq  \tilde{T}_{x,m} \left( g(x,e) \right).
$$
Because $\tilde{T}_{x',n}(y)$ is strictly increasing in $y$, there exists the inverse function $\tilde{T}_{x',n}^{-1} : \mathcal{Y}_{\pi^n(x')} \rightarrow \mathcal{Y}_{x'}$.
We define $\tilde{T}_{x',n}^{+} : \mathcal{Y}_{\pi^m(x)} \rightarrow \mathbb{R}$ as
$$
\tilde{T}_{x',n}^{+}(y) = \begin{cases}
    \tilde{T}_{x',n}^{-1}(y), & \text{if $y \in \mathcal{Y}_{\pi^n(x')}$}  \\
    \overline{y}_{x'}, & \text{otherwise}
  \end{cases}.
$$
Then, for all $(n,m) \in \Pi_{x',x}^M$ and $e \in (0,1)$, we obtain
$$
g(x',e) \leq \tilde{T}_{x',n}^{+} \left( \tilde{T}_{x,m} \left( g(x,e) \right) \right).
$$
We define
\begin{eqnarray}
T^{MU*}_{x',x}(y) &\equiv & \min_{(n,m) \in \Pi_{x',x}^M} \tilde{T}_{x',n}^{+} \left( \tilde{T}_{x,m} \left( g(x,e) \right) \right), \nonumber \\
G^{ML*}_x(u) &\equiv & \int F_{Y|X=x'}\left( T^{MU*}_{x',x}(u) \right) dF_X(x'). \nonumber
\end{eqnarray}
Then, $T^{MU*}_{x',x}(y)$ satisfies $g(x',e) \leq T^{MU*}_{x',x}\left(g(x,e)\right)$, but $G^{ML*}_x(u)$ may not be strictly increasing.
Hence, the upper bound of $g(x,e)$ is obtained from
$$
B^{ML*}(x,e) \equiv \sup_{y:y\leq x} \left\{ \inf\{u:G^{ML*}_y(u) \geq e\} \right\} \vee \underline{y}_x.
$$
Similarly, we can obtain the lower bound of $g(x,e)$ without Assumption 4 (ii).

\clearpage

\bibliographystyle{ecta}
\bibliography{nonseparable_model}

\end{document}